\documentclass[journal]{IEEEtran}
\usepackage{amsmath,amssymb,amsfonts,dsfont}

\usepackage{hyperref}
\usepackage[nocompress]{cite}

\usepackage[linesnumbered,ruled]{algorithm2e}
\usepackage{algorithmic,float}
\usepackage{setspace}

\usepackage{graphicx}
\makeatletter
    \let\MYcaption\@makecaption
\makeatother
    \usepackage[font=footnotesize]{subcaption}
\makeatletter
    \let\@makecaption\MYcaption
\makeatother

\usepackage{tikz}
\usetikzlibrary{shapes, shapes.geometric, arrows.meta, positioning, intersections, calc}
\tikzset{>=latex}

\usepackage{csquotes}
\renewcommand{\mkbegdispquote}[2]{\itshape}

\newcommand{\ubar}[1]{\text{\b{$#1$}}}

\DeclareMathOperator*{\argmin}{arg\,min}

\newtheorem{theorem}{Theorem}
\newtheorem{lemma}{Lemma}
\newtheorem{remark}{Remark}
\newtheorem{definition}{Definition}
\newtheorem{proposition}{Proposition}

\newtheorem{assumption}{Assumption}

\allowdisplaybreaks

\begin{document}
\title{Exploiting Data Significance in Remote Estimation of Discrete-State Markov Sources}
\author{Jiping~Luo,~Nikolaos~Pappas,~\IEEEmembership{Senior~Member,~IEEE}
\thanks{This work has been supported in part by the Swedish Research Council (VR), ELLIIT, the Graduate School in Computer Science (CUGS), the European Union (ETHER, 101096526, ROBUST-6G, 101139068, and 6G-LEADER, 101192080), and the European Union's Horizon Europe research and innovation programme under the Marie Skłodowska-Curie Grant Agreement No 101131481 (SOVEREIGN). An earlier version of this paper was presented in part at the 2024 ACM MobiHoc~\cite{jipingMobiHoc2024}. \textit{(Corresponding author: Jiping~Luo.)}}
\thanks{The authors are with the Department of Computer and Information Science, Link\"oping University, Link\"oping 58183, Sweden (e-mail: jiping.luo@liu.se; nikolaos.pappas@liu.se).}
}

\maketitle

\begin{abstract}
    We consider semantics-aware remote estimation of a discrete-state Markov source with both normal (low-priority) and alarm (high-priority) states. Erroneously announcing a normal state at the destination when the source is actually in an alarm state (i.e., missed alarm) incurs a significantly higher cost than falsely announcing an alarm state when the source is in a normal state (i.e., false alarm). Moreover, consecutive estimation errors may cause significant lasting impacts, such as maintenance costs and misoperations. Motivated by this, we introduce two new metrics, the Age of Missed Alarm (AoMA) and the Age of False Alarm (AoFA), to capture the lasting impacts incurred by different estimation errors. Notably, these two age processes evolve interdependently and distinguish between different error types. Our goal is to design a transmission policy that achieves an optimized trade-off between lasting impact and communication cost. The problem is formulated as a countably infinite-state Markov decision process (MDP) with an unbounded cost function. We show the existence of a simple switching policy with distinct thresholds for each age process and derive closed-form expressions for its performance. For symmetric and non-prioritized sources, we show that the optimal policy reduces to a threshold policy with identical thresholds. For numerical tractability, we propose a finite-state approximate MDP and prove that it converges exponentially fast to the original MDP in the truncation size. Finally, we develop an efficient search algorithm to compute the optimal switching policy and validate our theoretical findings with numerical results.
\end{abstract}

\section{Introduction}
\IEEEPARstart{E}{fficient} remote state estimation is the key to various networked control systems (NCSs), such as environmental monitoring, smart factories, swarm robotics, and connected autonomous vehicles~\cite{hespanha2007survey, park2018wireless, luo2023real}. The remote estimation problem involves a sensor sending status updates about an information source to a remote receiver over a wireless channel. In light of practical constraints, such as battery capacity, bandwidth limit, and network contention, a fundamental question is how to achieve an optimal trade-off between estimation quality and communication cost~\cite{baillieul2007control, ling2011sensor, chen2017event, chakravorty2017fundamental, luo2024semantic}. A promising approach is to exploit the \emph{semantics} of information, in terms of significance, to reduce the amount of uninformative data transmissions~\cite{kountouris2021semantics, gielis2022critical, luo2024semantic, luo2025cost}.

Accuracy and freshness are two dominant information attributes in the remote estimation literature. Accuracy is measured by distortion metrics such as mean-square error (MSE) and Hamming distortion, while freshness is typically captured by the notion of Age of Information (AoI), defined as the time elapsed since the most recently received update was generated (see, e.g.,~\cite{kosta2017age, sun2019age, yates2021age}). In some systems, distortion and AoI are closely related. For example, the minimum MSE of a linear time-invariant (LTI) system is a monotonically nondecreasing function of AoI~\cite{wu2020learning, ayan2019age, jiping2025wearing}. In such systems, AoI serves as a sufficient statistic for decision-making, and the monotonicity implies the existence of threshold-type optimal transmission policies. In general, however, the relationship between age and distortion is not well understood, thus it is interesting to be further investigated~\cite{inan2021optimal, bastopcu2021age, jayanth2023distortion}. 

This paper aims to answer the following questions:
\begin{displayquote}
    Does more accurate or fresher mean more valuable? Are all states equally important?
\end{displayquote}
We primarily focus on the remote estimation of a discrete-state Markov source with prioritized states, as depicted in Fig.~\ref{fig:system-model}. We first explain that distortion and AoI become inefficient in such systems and then introduce a new metric.

\begin{figure}[t]
    \centering
    \includegraphics[width=1\linewidth]{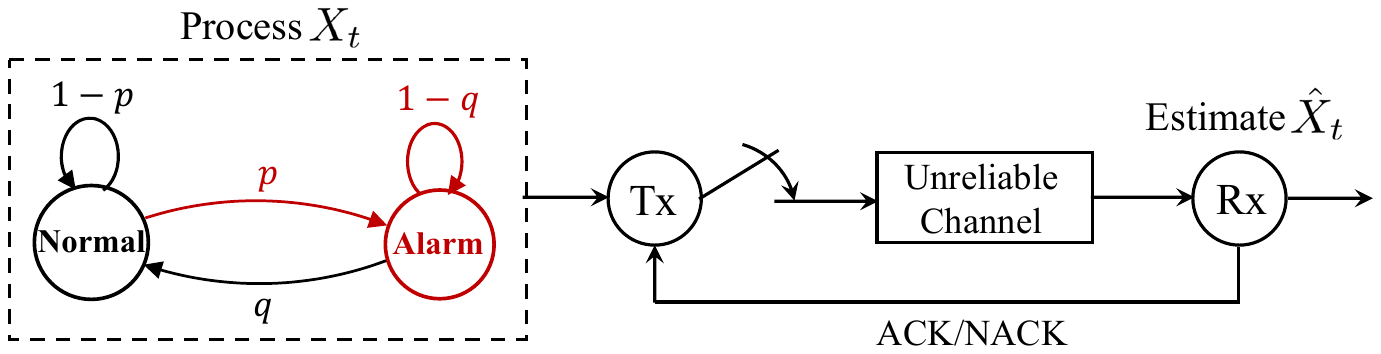}
    \caption{Remote estimation of a Markov source with normal and alarm states.}
    \label{fig:system-model}
\end{figure}

A notable difference between Markov and LTI models is that the estimation error does not necessarily evolve monotonically with time. AoI becomes inefficient in such systems as it ignores the source dynamics. Recent studies have proposed alternative metrics to address this issue. The Age of Incorrect Information (AoII) (see e.g.,~\cite{maatouk2020ToN, maatouk2023TWC, chen2024minimizing}) and the cost of memory error~\cite{salimnejad2024real} penalize the system through cost functions that scale with the time it remains erroneous. Notably, these metrics capture the semantic attribute of \emph{lasting impact} in consecutive errors, which is relevant in many NCSs. For example, successive reception of erroneous status updates from a remotely controlled drone can lead to wrong operations or even crashes. Moreover, studies in~\cite{champati2022detecting, salimnejad2024version} evaluated freshness in terms of content updates rather than timestamps. The intuition is that the information age need not grow when the content of the source remains unchanged. However, these metrics might not suffice in our problem, as they treat all source states equally, i.e., content-agnostic. This egalitarianism results in inadequate transmissions in alarm states but excessive transmissions in normal states. Intuitively, a remotely controlled drone should update its status more frequently in urgent situations, e.g., off-track or near obstacles. 

There have been efforts to exploit the semantic attribute of \emph{content awareness} in remote estimation systems. For instance, \cite{stamatakis2019control} assigns a quadratic age variable for the alarm state and a linear age variable for the normal state, thus imposing a higher cost for outdated alarm information. The cost of actuation error (CAE)~\cite{pappas2021goal} measures the cost of estimation error based on its contextual relevance and potential control risks to system performance. Similar considerations have been applied in the analysis of control stability~\cite{kriouile2025semantics}. Content-agnostic and content-aware randomized policies have been studied in CAE performance evaluation~\cite{salimnejad2023state, salimnejad2024real}. These policies find application in various remote tracking systems (see, e.g.,~\cite{nikkhah2023age, akar2024query, bastopcu2022using}). Optimal and learning-based policies in multi-source systems were derived in~\cite{luo2024semantic} and~\cite{luo2024goal}.

The most relevant studies to this paper are~\cite{jipingMobiHoc2024} and~\cite{bastopcu2022using}. \cite{jipingMobiHoc2024} is an earlier conference version of the present work, in which we introduced the \emph{Age of Missed Alarm} (AoMA) and the \emph{Age of False Alarm} (AoFA) to capture the lasting impacts of missed alarms and false alarms, respectively. Notably, AoMA and AoFA evolve interdependently, allowing them to distinguish between different estimation errors and synced states, thus content-awareness. We formulated an optimization problem to achieve an optimized tradeoff between age penalty and communication cost. The optimal policy was derived through the lens of Markov decision theory, and some preliminary results were reported without proof. Similar metrics were also studied in the context of health monitoring in~\cite{bastopcu2022using}, where suboptimal content-aware randomized policies were examined.

Our main contributions are as follows: (1) The optimal transmission problem is formulated as a Markov decision process (MDP) with an average-cost criterion and a countably infinite state space. We show that the optimal policy has a \emph{switching structure}; that is, the sensor triggers transmission in certain synced states or when the AoMA or AoFA exceeds a fixed threshold. (2) We derive closed-form expressions for the performance of both switching and threshold policies. Building on these results, we show that when the source is symmetric and the states are equally important, the optimal policy degenerates into a simple \emph{threshold policy} with identical thresholds. (3) For numerical tractability, we truncate the age processes and introduce a finite-state approximate MDP. We show the asymptotic optimality of the truncated MDP and propose an efficient policy search algorithm that circumvents the complexities of classic dynamic programming methods. Numerical results are presented to validate the performance of the optimal switching policy.

\textit{Organization:} Sections~\ref{sec:system-model} and~\ref{sec:problem-formulation} present the system model and the problem formulation. Section~\ref{sec:structural-results} presents our main results on the optimal policy. Section~\ref{sec:solution-approach} shows the asymptotic optimality of the truncated problem. Finally, Sections~\ref{sec:numerical-results} and~\ref{sec:conclusion} present numerical results and conclusions.

\section{System Model}\label{sec:system-model}
\subsection{Remote Estimation Model}
We consider a remote estimation system depicted in Fig.~\ref{fig:system-model}. The model consists of the following ingredients:

\subsubsection{Source} 
The information source is modeled as a binary Markov chain $\{X_t\}_{t\geq 1}$. At each time slot $t$, the source is in either the normal (low-priority) state $0$ or the alarm (high-priority) state $1$. The slot length is defined as the interval between two successive state changes. The states may represent either quantization levels of a physical process or abstract statuses (e.g., operating modes, component failures, or abrupt changes) of a dynamic system~\cite{costa2005mjls}. 

The state transition probability matrix of the chain is
\begin{align}
    Q = \begin{bmatrix}
        \bar{p} & p\\
        q & \bar{q}
    \end{bmatrix}\notag
\end{align}
where $\bar{p} = 1-p$, $\bar{q} = 1-q$, and $Q_{i,j} = \Pr[X_{t+1} = j|X_t = i]$ is the probability of transitioning from state $i$ to state $j$ between two consecutive time slots. To avoid pathological cases, we assume $Q$ is \emph{irreducible}, i.e., $0<p,q<1$. Then the chain admits a stationary distribution $(\nu_i: i \in \{0, 1\})$, where $\nu_0 = \frac{q}{p+q}$ and $\nu_1 = \frac{p}{p+q}$ represent the average proportion of time the chain is in state $0$ and $1$, respectively. The chain is called \emph{symmetric} if $p=q$. We call a chain \emph{positively correlated} if the source evolution has positive autocorrelation, i.e., $p < \bar{q}$. The chain is negatively correlated if $p>\bar{q}$ and i.i.d. if $p = \bar{q}$.

\subsubsection{Channel} We consider an unreliable channel $\{H_t\}_{t\geq 1}$ modeled as an i.i.d. Bernoulli process with mean $p_s$, where $H_t =1$ denotes a successful transmission at time $t$ and $H_t = 0$ denotes a failed transmission. Upon successful reception, the receiver sends an acknowledgment (ACK) to the sensor; otherwise, it sends a negative ACK (NACK) to indicate failure. ACK and NACK packets are assumed error-free.

\subsubsection{Sensor and Receiver}
The sensor observes the source and decides when to transmit a new measurement. Let $A_t\in\mathcal{A}$ denote the decision variable, where $A_t = 1$ means transmission while $A_t = 0$ means no transmission. We assume that a decision made at time $t$ takes effect at the beginning of the next slot, not immediately; that is, if $A_t = 1$, a new packet $X_{t+1}$ is sampled and transmitted at the beginning of slot $t+1$. This action delay accounts for several practical considerations such as measurement delay and processing time~\cite{nilsson1998stochastic, katsikopoulos2003markov}, which can be more prominent in the case of very short duration of transmission slot that is common in 5G and beyond systems. 

We assume that the receiver does not know the source pattern and updates its estimate using the latest received measurement~\cite{maatouk2020ToN, luo2024semantic}, i.e., $\hat{X}_t = X_{G_t}$, where $G_t$ is the generation time of the latest received update. 

\begin{remark}
    The AoI, defined as $\Delta^\text{AoI}_t = t - G_t$, ignores the source dynamics and the information content. Thus, a large AoI does not imply poor estimation performance, as the system may remain synced for some time, and vice versa.
\end{remark}

The information available at the sensor at time $t$ is 
\begin{align}
    I_t = \{X_{1:t}, \hat{X}_{1:t}, A_{1:t-1}\}.\notag
\end{align}
The sensor chooses an action $A_t$ by a \emph{transmission rule} $\pi_t$, 
\begin{align}
    A_t = \pi_t(I_t) = \pi_t(X_{1:t}, \hat{X}_{1:t}, A_{1:t-1}).\label{eq:policy}
\end{align}
Notice that $\pi_t$ can be either deterministic, selecting an action in a given state with certainty, or randomized, specifying a probability distribution on the action space. The collection $\pi=\{\pi_t\}_{t=1}^{\infty}$ is called the \emph{transmission policy}.

\subsection{Content-Aware Age Processes and Performance Measure} 
Since the source states are not equally important, we distinguish the following types of estimation errors:
\begin{itemize}
    \item \textbf{False Alarm (FA):} It occurs when $X_t = 0$ and $\hat{X}_t = 1$, indicating an unnecessary alarm triggered by the receiver. While somewhat less important, FA errors may incur extra operation or maintenance costs.
    \item \textbf{Missed Alarm (MA):} It occurs when $X_t = 1$ and $\hat{X}_t = 0$, indicating a failure to detect an alarm by the receiver. MA errors are more crucial and thus higher penalties.
\end{itemize}
Prolonged persistence in an error state may lead to substantial ramifications beyond merely accumulated costs during that period. We introduce the AoFA and AoMA to quantify, respectively, the lasting impacts of consecutive FA and MA errors. The age processes are recursively defined as
\begin{align}
    \Delta_{t+1}^\text{FA} &:= 
    \begin{cases}
        \Delta_{t}^\text{FA} + 1, &X_{t+1} = 0, \hat{X}_{t+1}=1,\\
        0, &\text{otherwise},
    \end{cases} \\
    \Delta_{t+1}^\text{MA} &:= 
    \begin{cases}
        \Delta_{t}^\text{MA} + 1, &X_{t+1} = 1, \hat{X}_{t+1}=0,\\
        0, &\text{otherwise}.
    \end{cases}
\end{align}
Let $S_t = (X_t, \hat{X}_t, \Delta_{t}^\text{MA}, \Delta_{t}^\text{FA})$ denote the system state. Since $\Delta_{t}^\text{FA}$ and $\Delta_{t}^\text{MA}$ cannot simultaneously be positive, they can be formed in a compact way as
\begin{align}
    \Delta_t = \mathds{1}\{(X_t,\hat{X}_t)\hspace{-0.2em}=\hspace{-0.2em}(0,1)\}\Delta_{t}^\text{FA} + \mathds{1}\{(X_t,\hat{X}_t)\hspace{-0.2em}=\hspace{-0.2em}(1,0)\}\Delta_{t}^\text{MA},\notag
\end{align}
where $\mathds{1}\{\cdot\}$ is the indicator function. For simplicity, in the rest of the paper, we shall use
\begin{align}
    S_t = (X_t, \hat{X}_t, \Delta_{t})
\end{align}
as the system state. The state space is given by
\begin{align}
    \mathcal{S} = \{(0, 0, 0), (1, 1, 0)\}\cup\{(0, 1, \delta), (1, 0, \delta):\delta\geq 1\},\notag
\end{align}
where $\mathcal{S}_\text{synced} = \{(0, 0, 0), (1, 1, 0)\}$, $\mathcal{S}_\text{FA} = \{(0, 1, \delta),\delta\geq 1\}$, and $\mathcal{S}_\text{MA} = \{(1, 0, \delta),\delta\geq 1\}$ are the set of synced states, the set of FA errors, and the set of MA errors, respectively. Note that $\mathcal{S}$ is a \emph{countably infinite} set, as the age processes may grow unboundedly. It is worth mentioning that $S_t = (\Delta_t^\text{MA}, \Delta_t^\text{FA})$ is not a good choice since $S_t = (0,0)$ causes ambiguity in the synced states. We will demonstrate the importance of non-ambiguity in Section~\ref{sec:analysis-switching-policy}.

The (per-stage) cost of being in state $S_t$ is defined as
\begin{align}
    c(S_t) := \beta \Delta_{t}^\text{MA} + (1-\beta) \Delta_{t}^\text{FA}, \label{eq:weighted-age}
\end{align}
where $\beta\in[0,1]$ represents the \emph{significance} of the MA error. The expected cost of taking action $A_t$ in state $S_t$ is given by 
\begin{align}
    c(S_t, A_t) = \mathbb{E}[c(S_{t+1})|S_t, A_t],\label{eq:expected-age-cost}
\end{align}
where the expectation is taken over the channel uncertainties, the source statistics, and the (possibly) random actions. 

\begin{figure}[t]
    \centering
    \includegraphics[width=\linewidth]{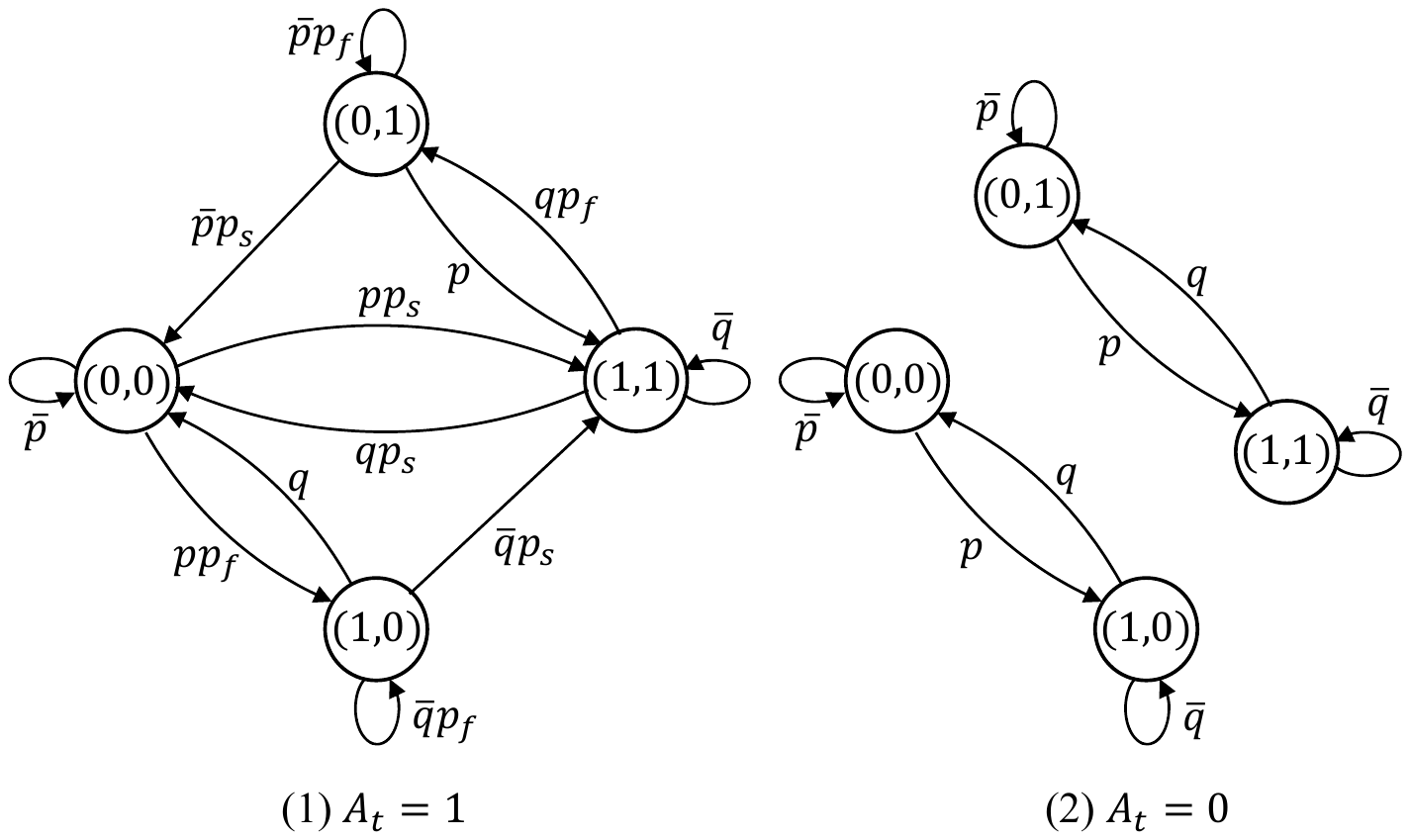}
    \caption{The evolution of estimation error under different actions.}
    \label{fig:state_DTMC_original}
\end{figure}

\subsection{System Dynamics}
Take action $a$ in state $s=(x, \hat{x},\delta)$, the system state will transition to state $s^\prime=(x^\prime, \hat{x}^\prime,\delta^\prime)$ with probability
\begin{align}
    P_{s,s^\prime}(a) &= \Pr[(x^\prime,\hat{x}^\prime,\delta^\prime)|(x,\hat{x},\delta),a]\notag\\
    &=\Pr[(x^\prime,\hat{x}^\prime)|(x,\hat{x}),a]
    \Pr[\delta^\prime|x^\prime, \hat{x}^\prime, \delta],\label{eq:system-dynamics}
\end{align}
where
\begin{align}
    &\Pr[(x^\prime,\hat{x}^\prime)|(x,\hat{x}),a]\notag\\
    =&\begin{cases}
        Q_{i,k}p_s, &a=1, (x,\hat{x})= 
        (i,j),(x^\prime,\hat{x}^\prime)=(k,k),k\neq j,\\ 
        Q_{i,k}p_f, &a=1, (x,\hat{x})= (i,j),(x^\prime,\hat{x}^\prime)=(k,j),k\neq j,\\ 
        Q_{i,j}, &a=1, (x,\hat{x})= (i,j),(x^\prime,\hat{x}^\prime)=(j,j),\\ 
        Q_{i,k}, &a=0, (x,\hat{x})= (i,j),(x^\prime,\hat{x}^\prime)=(k,j),\\
        0, &\text{otherwise},
    \end{cases}\notag
\end{align}
$p_f = 1-p_s, i,j,k\in\{0,1\}$, and
\begin{align}
    \Pr[\delta^\prime|x^\prime, \hat{x}^\prime, \delta] &= \begin{cases}
        1, &x^\prime\neq\hat{x}^\prime, \delta^\prime=\delta+1,\\
        1, &x^\prime=\hat{x}^\prime, \delta^\prime=0,\\
        0, &\text{otherwise}.
        \end{cases}\notag
\end{align}

The evolution of estimation error is depicted in Fig.~\ref{fig:state_DTMC_original}. The chain under the always-transmit policy (i.e., $A_t = 1, \forall t\geq 1$) is irreducible, whereas the non-transmit policy (i.e., $A_t = 0, \forall t\geq 1$) yields isolated groups of states. Hence, the underlying system with dynamics~\eqref{eq:system-dynamics} is multichain~\cite{puterman1994markov}. 

\begin{figure*}[t]
    \begin{subfigure}{0.28\linewidth}
        \centering
        \includegraphics[width=\linewidth]{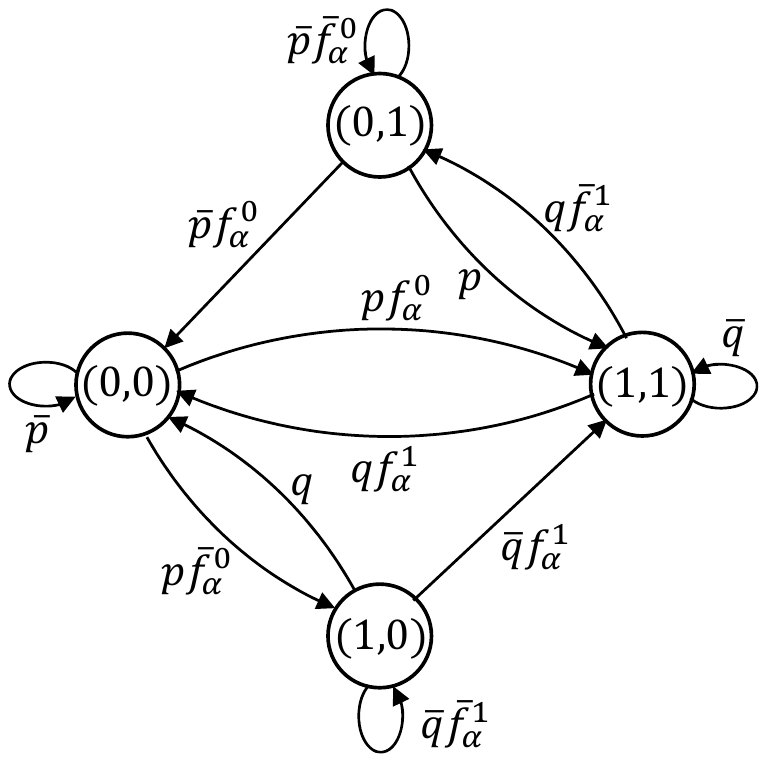}
        \caption{MC1 $\{(X_t, \hat{X}_t)\}_{t\geq 1}$}
        \label{fig:state_DTMC}
    \end{subfigure}
    \hfill
    \begin{subfigure}{0.71\linewidth}
        \centering
        \includegraphics[width=1\linewidth]{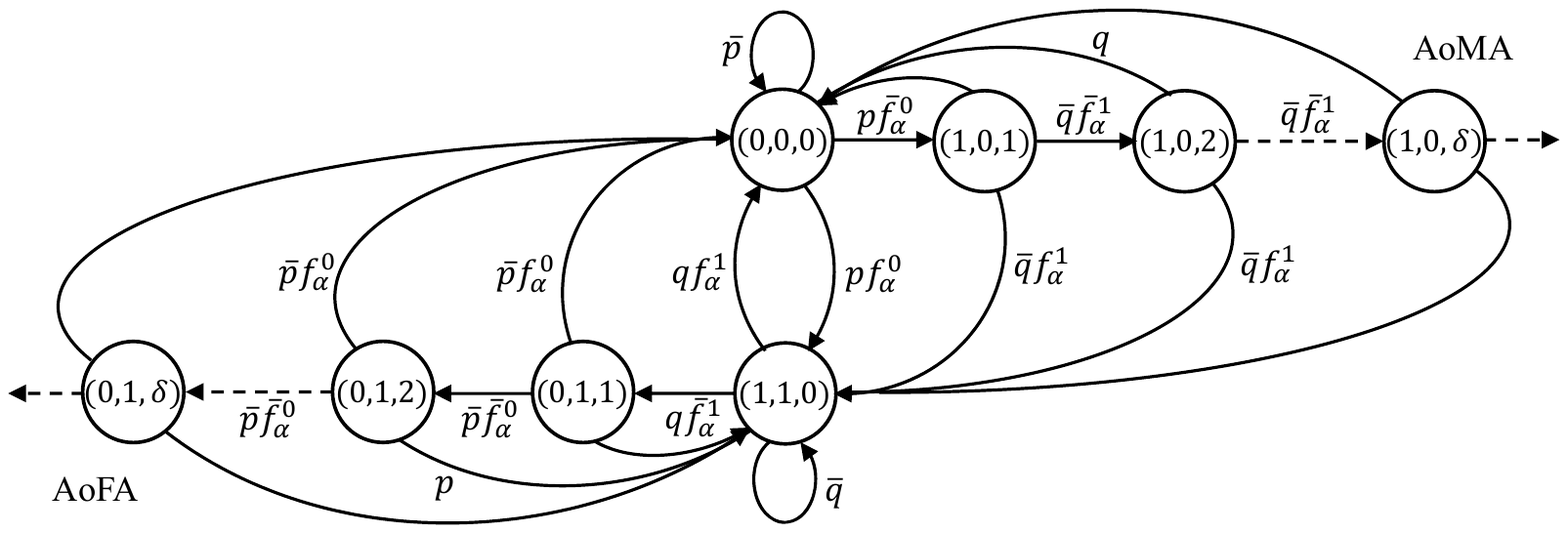}
        \caption{MC2 $\{(X_t, \hat{X}_t, \Delta_t)\}_{t\geq 1}$}
        \label{fig:joint_status_DTMC}
    \end{subfigure}     
    \caption{State evolution of the Markov chains induced by the age-agnostic randomized policy.}
    \label{fig:state-evolution}
\end{figure*}

\section{Problem Formulation and Analysis}\label{sec:problem-formulation}
\subsection{Problem Formulation}
Given a transmission policy $\pi$, the average semantics-aware estimation cost is defined as 
\begin{align}
	C(\pi) := \limsup_{T\rightarrow\infty}\frac{1}{T}\sum_{t=1}^{T}\mathbb{E}^{\pi} \bigl[c(S_t, A_t)\big|S_1 = s_1\bigr],
\end{align}
where $\mathbb{E}^\pi$ represents the conditional expectation, given that policy $\pi$ is employed with initial state $s_1 = (0,0,0)$. The transmission frequency is defined as
\begin{align}
	F(\pi) := \limsup_{T\rightarrow\infty}\frac{1}{T}\sum_{t=1}^{T}\mathbb{E}^\pi \bigl[f(S_t, A_t)\big|S_1 = s_1\bigr],
\end{align}
where $f(S_t, A_t) = \mathds{1}\{A_t\neq 0\}$. Given the communication cost (i.e., the cost of joint sampling and transmission) $\lambda$ for each transmission, the system's performance is measured by
\begin{align}
    \mathcal{L}(\pi) &:= \limsup_{T\rightarrow\infty}\frac{1}{T}\sum_{t=1}^{T}\mathbb{E}^\pi \bigl[\ell(S_t, A_t)\big|S_1 = s_1\bigr]\notag\\
    &= C(\pi) + \lambda F(\pi),\label{eq:total-cost}
\end{align}
where $\ell(S_t, A_t) = c(S_t, A_t) + \lambda f(S_t, A_t)$.

The sensor seeks an optimal policy $\pi^*$ that attains
\begin{align}
    \mathcal{L}^* = \inf_{\pi\in\Pi} \mathcal{L}(\pi),\label{problem:costly}
\end{align}
where $\Pi$ denotes the set of all admissible policies given in~\eqref{eq:policy}.

Problem \eqref{problem:costly} is an MDP characterized by $(\mathcal{S}, \mathcal{A}, P, \ell)$. However, it encounters computing and memory challenges due to the infinite state space and (possibly) unbounded per-stage costs. The main approach to remedy these difficulties is to reveal favorable properties of the optimal policy, thereby restricting the policy searching space. Another difficulty lies in the \textit{interdependence} between AoMA and AoFA (see Fig.~\ref{fig:state-evolution}). This dependency makes it challenging to analyze the limiting behavior of the underlying system.

We use the following quantities to measure the performance loss incurred by any policy $\pi\in\Pi$.

\begin{definition}
    The performance gap of a policy $\pi$ relative to the optimal policy $\pi^*$ is defined as $g(\pi, \pi^*) := \mathcal{L}(\pi) - \mathcal{L}^*$.
\end{definition}

\begin{definition}
    Let $\nu^\pi$ and $\nu^*$ denote the stationary distributions of the system $\{S_t\}$ induced by $\pi$ and $\pi^*$. The policy distance is defined as the Kullback-Leibler (KL) divergence between $\nu^\pi$ and $\nu^*$, i.e., $d(\nu^\pi, \nu^*) 
        := \sum_{s\in\mathcal{S}}\nu^\pi_s\log\frac{\nu^\pi_s}{\nu^*_s}$.
\end{definition}

\begin{remark}
    Note that $d(\nu^\pi, \nu^*) = 0$ implies $g(\pi, \pi^*) = 0$, but the reverse need not hold. That is, two policies may have comparable costs yet exhibit very different behaviors.
\end{remark}

\subsection{Age-Agnostic Randomized Policy}
To better illustrate the properties and evolution of AoMA and AoFA, we examine the content-aware randomized policies, which transmit at each time $t$ with a fixed probability $f_i$ when the source is in state $X_t = i$, where $\sum_{i}f_i\leq 1$. The content-agnostic randomized policies are special cases when $f_0 = f_1$. Since these policies do not depend on the age processes, we refer to them as \emph{age-agnostic} policies in the sequel. Fig.~\ref{fig:state-evolution} depicts the state evolution of the Markov chains $\{(X_t, \hat{X}_t)\}_{t\geq1}$ (MC1) and $\{(X_t, \hat{X}_t, \Delta_t)\}_{t\geq1}$ (MC2) induced by age-agnostic policies, where $f^i_\alpha = p_s f_i$ and $\bar{f}^i_\alpha = 1-p_s f_i$.

The following results are instrumental in establishing the existence of an optimal policy. Assertion (ii) implies that Problem~\eqref{problem:costly} is communicating, while assertions (iii-iv) imply the solvability of the problem.

\begin{proposition}\label{proposition:sa-policy}
    The chains MC1 and MC2 induced by the age-agnostic policy $\pi = (f_0, f_1)$ satisfy:
    \begin{itemize}
        \item[(i)] MC1 is a finite-state ergodic Markov chain and admits the following stationary distribution
        \begin{subequations}
        \begin{align}
        \nu_{0,0} &= q(\bar{p}f^0_\alpha+pf^1_\alpha)(q+\bar{q}f^1_\alpha)/\zeta,\\
        \nu_{0,1} &= pq\bar{f}^1_\alpha(qf^0_\alpha + \bar{q}f^1_\alpha)/\zeta,\\
        \nu_{1,0} &= pq\bar{f}^0_\alpha(\bar{p}f^0_\alpha + pf^1_\alpha)/\zeta,\\
        \nu_{1,1} &= p(qf^0_\alpha+\bar{q}f^1_\alpha)(p+\bar{p}f^0_\alpha)/\zeta,
        \end{align}\label{eq:stationary-distribution-sa}
        \end{subequations}
        where $\zeta = (p+q)\big(qf^0_\alpha+pf^1_\alpha + (1-p-q)f^0_\alpha f^1_\alpha\big)$ and $\nu_{i,j}$ represents the average proportion of time MC1 spends in state $(i,j)$.
        \item[(ii)] MC2 is a countable-state, irreducible, and positive recurrent Markov chain with stationary distribution
        \begin{subequations}\label{eq:main}
        \begin{align}
            \nu_{0,0,0} &= \nu_{0,0}, ~\nu_{1,1,0} = \nu_{1,1}, \label{eq:sa-synced}\\
            \nu_{1,0,k} &= p\bar{f}^0_\alpha(\bar{q}\bar{f}^1_\alpha)^{k-1}\nu_{0,0}, \,k\geq 1\label{eq:sa-MA}, \\
            \nu_{0,1,k} &= q\bar{f}^1_\alpha(\bar{p}\bar{f}^0_\alpha)^{k-1}\nu_{1,1},\, k\geq 1, \label{eq:sa-FA}
        \end{align}\label{eq:sa-stationary}
        \end{subequations}
        where $\nu_{i,j,\delta}$ represents the average proportion of time MC2 spends in state $(i,j,\delta)$. The average cost is 
        \begin{align}
    \hspace{-0.5em}\mathcal{L}(\pi) 
        \hspace{-0.2em}=\hspace{-0.2em} \frac{\beta p\bar{f}^0_\alpha\nu_{0,0}}{(1-\bar{q}\bar{f}^1_\alpha)^2} \hspace{-0.2em}+\hspace{-0.2em} \frac{(1\hspace{-0.2em}-\hspace{-0.2em}\beta)\bar{f}^1_\alpha\nu_{1,1}}{(1-\bar{p}\bar{f}^0_\alpha)^2} \hspace{-0.2em}+\hspace{-0.2em} \frac{\lambda(qf_0+pf_1)}{p+q}.\label{eq:sa-cost}
        \end{align}
        \item[(iii)] There exists a state $z$ in MC2 such that the expected first passage time from any state $i\in \mathcal{S}\backslash\{z\}$ to $z$ is finite. 
        \item[(iv)] For state $z$ in (iii), the expected first passage cost from any state $i\in \mathcal{S}\backslash\{z\}$ to $z$ is also finite. 
    \end{itemize}
\end{proposition}
\begin{IEEEproof}
    (Sketch) The key step is to write the balance equations for the stationary distributions of the chains induced by age-agnostic policies. See Appendix~\ref{proof:sa-policy} for details.
\end{IEEEproof}

\section{Structural Results of the Optimal Policy}\label{sec:structural-results}
\subsection{Optimality of Switching Policy}
This section presents the main results on the optimal policy. We first show that there is no loss of optimality in restricting attention to stationary and deterministic policies.

\begin{proposition}\label{proposition:existence-of-lambda-policy} There exists a bounded function $h$ such that
\begin{align}
    \mathcal{L}^* + h(i) = \min\limits_{a}\{\ell(i,a)+ \sum\nolimits_{j}P_{i,j}(a)h(j)\}\label{problem:bellman-equation}
\end{align}
for all $i\in\mathcal{S}$. The optimal policy $\pi^*$ is deterministic, given by
    \begin{align}
        \pi^*(i) = \argmin\limits_{a}\{\ell(i,a)+ \sum\nolimits_{j}P_{i,j}(a)h(j)\},i\in\mathcal{S}.\label{eq:lamba-optimal-policy}
    \end{align}
\end{proposition}
\begin{IEEEproof}
    We rely on \cite[Corollary~7.5.10]{sennott1998stochastic} and verify that the following conditions are valid:
    \begin{enumerate}
        \item[C1:] There exists a stationary policy $\pi^\prime$ and a state $z\in\mathcal{S}$ such that MC2 has finite expected first passage time and cost from every other state to $z$;
        \item[C2:] MC2 induced by policy $\pi^\prime$ is positive recurrent;
        \item[C3:] Given any positive constant $U$, $\mathcal{S}_U := \{s\in\mathcal{S}|\ell(s,a)\leq U~\text{for some}~a\}$ is finite.
    \end{enumerate}
    C1 and C2 follow from Proposition~\ref{proposition:sa-policy}. C3 can be verified via the always-transmit policy.
\end{IEEEproof}

\begin{remark}
    One may try to apply dynamic programming methods to solve Bellman's equation \eqref{problem:bellman-equation}. However, it is not applicable in practice as we cannot iterate over infinitely many states. Furthermore, even with such an optimal policy at hand, it is impossible to store all state-action pairs due to memory constraints at the sensor. Deep reinforcement learning (DRL) techniques can partially address these challenges, however, at the cost of suboptimality and lack of interpretability.
\end{remark}

The following theorem shows that the optimal policy exhibits a simple \emph{switching structure}, which facilitates policy storage and algorithm design. Notably, due to the consideration of action delay, the best strategy may trigger transmission in the synced states. For example, if the sensor transmits in the synced state $(0,0,0)$, there is a high probability that the system will directly transition to synced state $(1,1,0)$, thereby forcing the system into low-cost FA errors. 

\begin{definition}\label{definition:state-ordering}
    We define the ordering $s_1 \leq s_2$ for $s_1, s_2 \in \{(0,0,0)\}\cup\mathcal{S}_{\rm{MA}}$ if their ages satisfy $\delta_1\leq \delta_2$. Similarly, define $s_1 \leq s_2$ if $s_1, s_2 \in \{(1,1,0)\}\cup\mathcal{S}_{\rm{FA}}$ and $ \delta_1\leq \delta_2$. 
\end{definition}

\begin{assumption}\label{assumption:source}
    The source is positively correlated or i.i.d. 
\end{assumption}

\begin{theorem}\label{theorem:structure-lambda-optimal}
    The optimal policy has a switching structure. That is, there exists two positive thresholds $\bar{\delta}^*, \ubar{\delta}^*\geq 1$ and two actions $\bar{a}^*, \ubar{a}^*\in\{0, 1\}$ such that
    \begin{align}
        \pi^*(s) = \begin{cases}
            1, &s\geq (1, 0, \bar{\delta}^*)~\text{or}~s\geq (0, 1, \ubar{\delta}^*),\\
            \bar{a}^*, &s = (0,0,0),\\
            \ubar{a}^*, &s = (1,1,0),\\
            0, &\text{otherwise},
        \end{cases}\label{eq:switching-policy}
    \end{align}
    Furthermore, under Assumption~\ref{assumption:source}, there exists two thresholds $\bar{\delta}^*, \ubar{\delta}^*\geq 0$ such that 
    \begin{align}
        \pi^*(s) = \begin{cases}
            1, &s \geq \bar{s}^*~\text{or}~s\geq \ubar{s}^*,\\
            0, &\text{otherwise},
        \end{cases}
    \end{align}
    where $\bar{s}^* = (1, 0, \bar{\delta}^*)$ if $\bar{\delta}^*\geq 1$ and $\bar{s}^* = (0, 0, 0)$ if $\bar{\delta}^*=0$, and $\ubar{s}^* = (0, 1, \ubar{\delta}^*)$ if $\ubar{\delta}^*\geq 1$ and $\ubar{s}^* = (1, 1, 0)$ if $\ubar{\delta}^*=0$.
\end{theorem}
\begin{IEEEproof} (Sketch) We rely on~\cite[Theorem~8.11.3]{puterman1994markov} to show that the function $\ell(i,a)+ \sum\nolimits_{j}P_{i,j}(a)h(j)$ in~\eqref{problem:bellman-equation} is submodular in age for each fixed error. See Appendix~\ref{poof:structure-lambda-optimal} for details.
\end{IEEEproof}

Under the mild Assumption~\ref{assumption:source}, finding the optimal policy reduces to determining two optimal threshold values $\bar{\delta}^*$ and $\ubar{\delta}^*$. Thus, the sensor only needs to store two threshold values instead of all possible state-action pairs, thereby remedying the ``curse of memory". Fig.~\ref{fig:threshold-policy} depicts the state evolution of MC2 induced by a switching policy with thresholds $\bar{\delta}, \ubar{\delta}\geq 1$. The sensor triggers transmission only when the AoMA exceeds $\bar{\delta}$ or the AoFA exceeds $\ubar{\delta}$.

\begin{figure}[t]
    \centering
    \includegraphics[width=\linewidth]{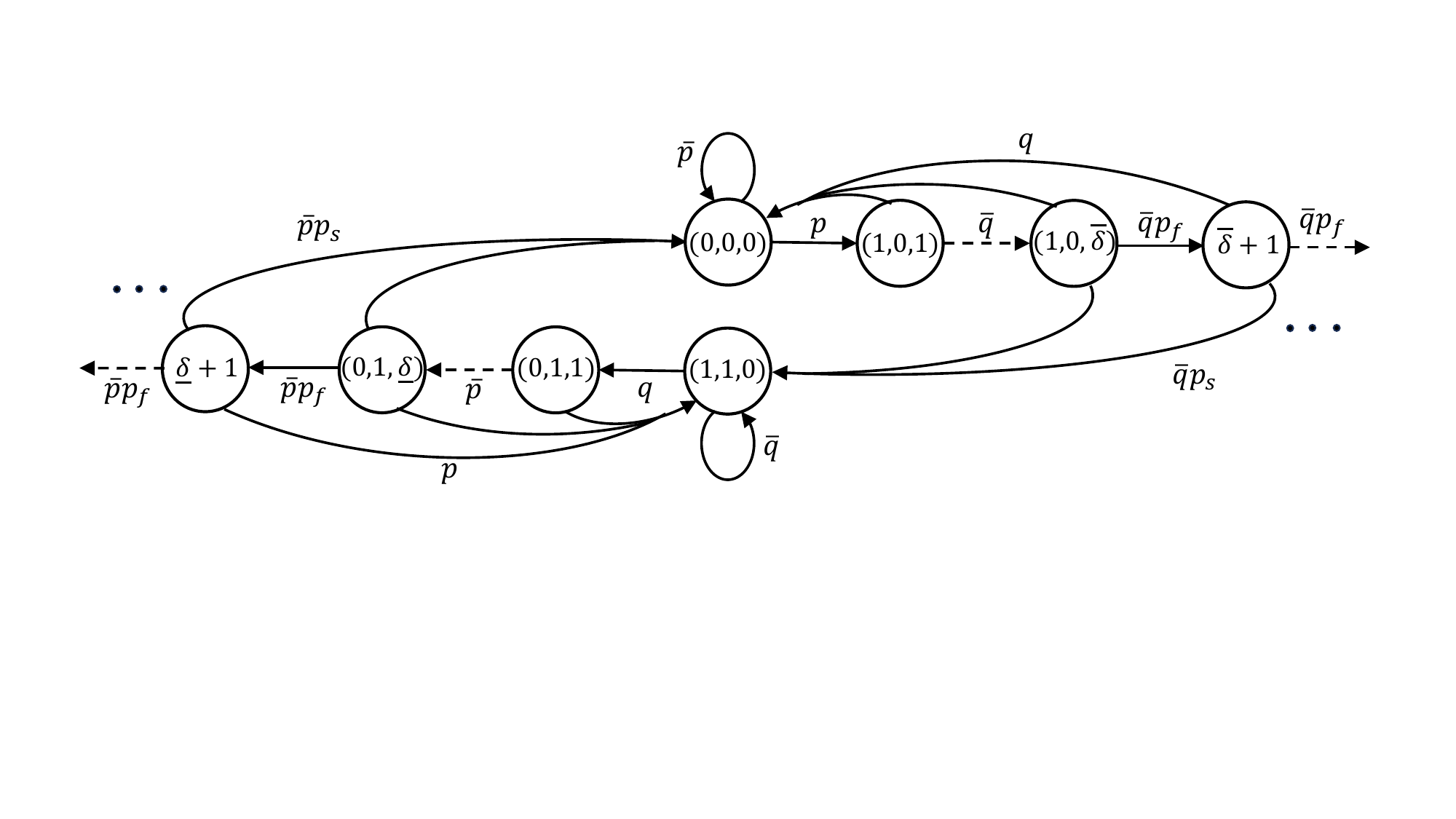}
    \caption{State evolution of MC2 induced by the switching policy.}
    \label{fig:threshold-policy}
\end{figure}

\subsection{Analysis of Switching Policy}\label{sec:analysis-switching-policy}
We proceed under Assumption~\ref{assumption:source} and distinguish between the following four cases: 
\begin{enumerate}
    \item[(i)] no transmission in synced states, i.e., $\bar{\delta},\ubar{\delta}\geq 1$;
    \item[(ii)] transmission in $(0,0,0)$, i.e., $\bar{\delta} = 0$ and $\ubar{\delta} \geq 1$;
    \item[(iii)] transmission in $(1,1,0)$, i.e., $\ubar{\delta} = 0$ and $\bar{\delta} \geq 1$;
    \item[(iv)] transmission in both synced states, i.e., $\bar{\delta} = \ubar{\delta} = 0$.
\end{enumerate}
When Assumption~\ref{assumption:source} does not hold, transmitting in the synced state $(0,0,0)$ does not imply transmission in all MA errors, as a positive threshold for AoMA may still exist. For example, case (iv) corresponds to transmitting in the synced states as well as in MA and FA errors whose age exceeds the thresholds $\bar{\delta}, \ubar{\delta} \geq 1$. Nevertheless, the analysis is analogous.

The main result of this section follows. Proposition~\ref{proposition:stationary-distribution-threshold} presents the stationary distribution of switching policies in the dominating case (i). Theorem~\ref{theorem:average-cost-threshold} derives closed-form expressions for the system's performance.

\begin{proposition}\label{proposition:stationary-distribution-threshold} 
    For a given switching policy $\pi$ with thresholds $\bar{\delta}, \ubar{\delta} \geq 1$, MC2 induced by $\pi$ is irreducible and admits a stationary distribution given by
    \begin{subequations}
    \begin{align}
        \nu_{0,0,0} &= \Gamma_0(a, b), \,\,\nu_{1,1,0} = \Gamma_1(a, b),\label{eq:stationary-threshold-a}\\
        \nu_{1,0,k} &= \begin{cases}
            p\bar{q}^{k-1}\nu_{0,0,0}, & 1\leq k\leq \bar{\delta},\\
            p\bar{q}^{\bar{\delta}-1}(\bar{q}p_f)^{k-\bar{\delta}}\nu_{0,0,0}, & k> \bar{\delta},
        \end{cases}\label{eq:stationary-threshold-c}\\
        \nu_{0,1,k} &= \begin{cases}
            q\bar{p}^{k-1}\nu_{1,1,0}, & 1\leq k\leq \ubar{\delta},\\
            q\bar{p}^{\ubar{\delta}-1}(\bar{p}p_f)^{k-\ubar{\delta}}\nu_{1,1,0}, & k> \ubar{\delta},\label{eq:stationary-threshold-d}
        \end{cases}
    \end{align}
    \end{subequations}
    where 
    \begin{gather*}
        \Gamma_0(a, b) = \frac{\bar{p}b(\ubar{\delta})}{(1+a(1))\bar{p}b(\ubar{\delta})+(1+b(1))\bar{q}a(\bar{\delta})},\\
        \Gamma_1(a, b) = \frac{\bar{q}a(\bar{\delta})}{(1+a(1))\bar{p}b(\ubar{\delta})+(1+b(1))\bar{q}a(\bar{\delta})},\\
        a(1) = \frac{p(1-\bar{q}^{\bar{\delta}-1})}{1-\bar{q}} + a(\bar{\delta}),\,\,\, a(\bar{\delta}) =\frac{p\bar{q}^{\bar{\delta}-1}}{1-\bar{q}p_f},\\
        b(1) =\frac{q(1-\bar{p}^{\ubar{\delta}-1})}{1-\bar{p}} + b(\ubar{\delta}), \,\,\, b(\ubar{\delta}) = \frac{q\bar{p}^{\ubar{\delta}-1}}{1-\bar{p}p_f}.
    \end{gather*}
\end{proposition}
\begin{IEEEproof} (Sketch) Write the balance equations for MC2 induced by switching policies and verify that they yield a unique stationary distribution. See Appendix~\ref{proof:stationary-distribution-threshold} for details.
\end{IEEEproof}

\begin{theorem}\label{theorem:average-cost-threshold}
    For a given switching policy $\pi$ with thresholds $\bar{\delta}, \ubar{\delta}\geq 1$, the average cost is given by
    \begin{align}
        \mathcal{L}(\pi) &= \beta p \psi(\bar{q}, \bar{\delta})\nu_{0,0,0} + (1-\beta) q  \psi(\bar{p}, \ubar{\delta})\nu_{1,1,0}\notag\\
        &\quad+ \lambda \big(a(\bar{\delta})\nu_{0,0,0} + b(\ubar{\delta})\nu_{1,1,0}\big),
    \end{align}
    where $a(\bar{\delta}), b(\ubar{\delta}),\nu_{0,0,0}$, and $\nu_{1,1,0}$ are given in Proposition~\ref{proposition:stationary-distribution-threshold}, and $\psi(x,y)=\frac{1-(x + (1-x)y)x^{y-1}}{(1-x)^2} + \frac{(x p_f + (1-xp_f)y)x^{y-1}}{(1-xp_f)^2}$.
\end{theorem}
\begin{IEEEproof}
    See Appendix~\ref{proof:average-cost-threshold}.
\end{IEEEproof}

The results for cases (ii)-(iv) are summarized in the following lemmas. Proofs are analogous and omitted for brevity.

\begin{lemma}\label{lemma:case01} 
    For a given switching policy $\pi$ with thresholds $\bar{\delta} = 0, \ubar{\delta}\geq 1$, the stationary distribution satisfies
    \begin{subequations}
    \begin{align}
        \nu_{0,0,0} &= \Gamma_0(a, b), \,\,\nu_{1,1,0} = \Gamma_1(a, b),\\
        \nu_{1,0,k} &= pp_f (\bar{q}p_f)^{k-1}, k\geq 1,\\
        \nu_{0,1,k} &= \begin{cases}
            q\bar{p}^{k-1}\nu_{1,1,0}, & 1\leq k\leq \ubar{\delta},\\
            q\bar{p}^{\ubar{\delta}-1}(\bar{p}p_f)^{k-\ubar{\delta}}\nu_{1,1,0}, & k> \ubar{\delta}.
        \end{cases}
    \end{align}
    \end{subequations}
    where $b(1)$ and $b(\ubar{\delta})$ are given in Proposition~\ref{proposition:stationary-distribution-threshold}, and $a(1) = \frac{pp_f}{1-\bar{q}p_f}, a(\bar{\delta}) =\frac{a(1)}{\bar{q}p_f}$. Furthermore, the average cost is given by 
    \begin{align}
        \mathcal{L}(\pi) &= \frac{\beta pp_f \nu_{0,0,0}}{(1-\bar{q}p_f)^2} + (1-\beta) q  \psi(\bar{p}, \ubar{\delta})\nu_{1,1,0}\notag\\
        &\quad + \lambda\big((1+a(1))\nu_{0,0,0} + b(\ubar{\delta})\nu_{1,1,0}\big).
    \end{align}
\end{lemma}

\begin{lemma}\label{lemma:case10} 
    For a given switching policy $\pi$ with thresholds $\bar{\delta} \geq 1, \ubar{\delta}=0$, the stationary distribution satisfies
    \begin{subequations}
    \begin{align}
        \nu_{0,0,0} &= \Gamma_0(a, b), \,\,\nu_{1,1,0} = \Gamma_1(a, b),\\
        \nu_{1,0,k} &= \begin{cases}
            p\bar{q}^{k-1}\nu_{0,0,0}, & 1\leq k\leq \bar{\delta},\\
            p\bar{q}^{\bar{\delta}-1}(\bar{q}p_f)^{k-\bar{\delta}}\nu_{0,0,0}, & k> \bar{\delta},
        \end{cases}\\
        \nu_{0,1,k} &=  qp_f (\bar{p}p_f)^{k-1}, k\geq 1.
    \end{align}
    \end{subequations}
    where $a(1)$ and $a(\bar{\delta})$ are the same as in Proposition~\ref{proposition:stationary-distribution-threshold}, $b(1) = \frac{qp_f}{1-\bar{p}p_f}$, and $b(\ubar{\delta}) =\frac{b(1)}{\bar{p}p_f}$. Furthermore, the average cost is
    \begin{align}
        \mathcal{L}(\pi) &= \beta p \psi(\bar{q}, \bar{\delta})\nu_{0,0,0} + \frac{(1-\beta) qp_f \nu_{1,1,0}}{(1-\bar{p}p_f)^2}\notag\\
        &\quad + \lambda\big(a(\bar{\delta})\nu_{0,0,0} + (1+b(1))\nu_{1,1,0}\big).
    \end{align}
\end{lemma}

\begin{lemma}\label{lemma:case00}
    For a given switching policy $\pi$ with thresholds $\bar{\delta}=\ubar{\delta} =0$, the stationary distribution satisfies
    \begin{subequations}
    \begin{align}
        \nu_{0,0,0} &= \Gamma_0(a, b), \,\,\nu_{1,1,0} = \Gamma_1(a, b),\\
        \nu_{1,0,k} &= pp_f(\bar{q}p_f)^{k-1}\nu_{0,0}, \,k\geq 1, \\
        \nu_{0,1,k} &= qp_f(\bar{p}p_f)^{k-1}\nu_{1,1},\, k\geq 1. 
    \end{align}
    \end{subequations}
    where $a(1) = \frac{pp_f}{1-\bar{q}p_f}, a(\bar{\delta}) =\frac{a(1)}{\bar{q}p_f},b(1) = \frac{qp_f}{1-\bar{p}p_f},b(\ubar{\delta}) =\frac{b(1)}{\bar{p}p_f}$. Furthermore, the average cost is given by
    \begin{align}
        \mathcal{L}(\pi) = \frac{\beta p p_f \nu_{0,0,0}}{(1-\bar{q}p_f)^2} + \frac{(1-\beta) q p_f \nu_{1,1,0}}{(1-\bar{p}p_f)^2} + \lambda.
    \end{align}
\end{lemma}

\subsection{Symmetric Source and Threshold Policy}\label{sec:symmetric-sources}

It is natural to ask:
\begin{displayquote}
    When is it optimal to use identical thresholds?
\end{displayquote}
Intuitively, for symmetric sources with equally important states, the sensor has no preference for different estimation errors, and MC2 induced by a switching policy with identical thresholds (i.e., threshold policy) is symmetric. Hence, it is natural to expect that threshold policies suffice in such cases. This intuition is justified in the following theorem.

\begin{assumption}\label{assumption:symmetric}
    The source is symmetric and non-prioritized, i.e., $p=q, \beta=0.5$.
\end{assumption}
\begin{theorem}\label{theorem:AoII-threshold}
    Under Assumption~\ref{assumption:symmetric}, the optimal policy has identical thresholds.
\end{theorem}
\begin{IEEEproof} 
    We prove the theorem using the analytical results derived in Section~\ref{sec:analysis-switching-policy}. For symmetric sources ($p=q$), the quantities $a(\cdot)$ and $b(\cdot)$ in Proposition~\ref{proposition:stationary-distribution-threshold} are identical. Define $\xi(x) = 1 - \bar{p}^{x-1}$ and $\psi(x) = \psi(\bar{p}, x)$. The average cost of any switching policy $\pi$ with thresholds $x,y\geq 1$ is given by
    \begin{align}
        \mathcal{L}(x,y)
        &\hspace{-0.3em}=\hspace{-0.3em}\frac{\big(\beta p \psi(x) \hspace{-0.2em}+\hspace{-0.2em} \lambda a(x)\big)a(y) \hspace{-0.2em}+\hspace{-0.2em} \big((1\hspace{-0.2em}-\hspace{-0.2em}\beta) p \psi(y) \hspace{-0.2em}+\hspace{-0.2em} \lambda a(y)\big)a(x)}{\big(1 \hspace{-0.2em}+\hspace{-0.2em} \xi(x) \hspace{-0.2em}+\hspace{-0.2em} a(x)\big)a(y) \hspace{-0.2em}+\hspace{-0.2em} \big(1 \hspace{-0.2em}+\hspace{-0.2em} \xi(y) \hspace{-0.2em}+\hspace{-0.2em} a(y)\big)a(x)}\notag\\
        &\hspace{-0.3em}=\hspace{-0.3em}\frac{\frac{\beta p \psi(x) + \lambda a(x)}{a(x)} + \frac{(1-\beta) p \psi(y) + \lambda a(y)}{a(y)}}{\frac{1 + \xi(x) + a(x)}{a(x)} + \frac{1 + \xi(y) + a(y)}{a(y)}}.\notag
    \end{align}
    For threshold policies (i.e., $y = x$), $\mathcal{L}(x,y)$ reduces to
    \begin{align*}
        \mathcal{L}(x) =\frac{0.5 p \psi(x) + \lambda a(x)}{1 + \xi(x) + a(x)}.
    \end{align*}
    Let $m_x = 0.5\, p \psi(x) + \lambda a(x)$ and
    $n_x = 1 + \xi(x) + a(x)$. We may write $\mathcal{L}(x) = m_x / n_x$ and
    \begin{align}
        \mathcal{L}(x,y) = \frac{m_x + m_y}{n_x + n_y} + \frac{(\beta-0.5)\big(\frac{\psi(x)}{a(x)} - \frac{\psi(y)}{a(y)}\big)}{\frac{1 + \xi(x) + a(x)}{a(x)} + \frac{1 + \xi(y) + a(y)}{a(y)}}.\notag
    \end{align}
    When the source is non-prioritized (i.e., $\beta = 0.5$), the above expression reduces to
    \begin{align*}
        \mathcal{L}(x,y) = \frac{m_x + m_y}{n_x + n_y}.
    \end{align*}
    
    Since $n_x, n_y > 0$, we have
    \begin{align*}
        \frac{m_x + m_y}{n_x + n_y} = \frac{\frac{m_x}{n_x}n_x + \frac{m_y}{n_y}n_y}{n_x + n_y} \geq \min\{\frac{m_x}{n_x}, \frac{m_y}{n_y}\},
    \end{align*}
    which implies that
    \begin{align*}
        \mathcal{L}(x, y) \geq \min\{\mathcal{L}(x), \mathcal{L}(y)\}
    \end{align*}
    for all $x, y\geq 1$. Therefore, under Assumption~\ref{assumption:symmetric}, using distinct thresholds cannot achieve a lower average cost than a single-threshold policy, which completes the proof.
\end{IEEEproof}

A similar result derived in~\cite{maatouk2020ToN} shows that the AoII-optimal policy for symmetric sources has a threshold structure. The proof techniques, however, differ from those used here.

\section{Asymptotic Optimality}\label{sec:solution-approach}
It is impractical to search for the optimal thresholds over an infinite state space. For numerical tractability, we truncate the state space and propose a finite-state approximate MDP. The truncated age processes evolve as follows:
\begin{align}
    \Delta_{t+1}^\text{FA}(N) &:= \begin{cases}
        \big(\Delta_t^\text{FA} + 1\big)_N, &X_{t+1} = 0, \hat{X}_{t+1} = 1,\\
        0, &\text{otherwise},
    \end{cases}\\
    \Delta_{t+1}^\text{MA}(N) &:= \begin{cases}
        \big(\Delta_t^\text{MA} + 1\big)_N, &X_{t+1} = 1, \hat{X}_{t+1} = 0,\\
        0, &\text{otherwise},
    \end{cases}
\end{align}
where $\big(x\big)_N = x$ if $x\leq N$ and $\big(x\big)_N = N$ if $x > N$. Thus, the system is confined within 
\begin{align}
    \mathcal{S}_N := \mathcal{S}_\text{synced}\cup \{(1,0,\delta), (0,1,\delta):1\leq \delta\leq N\}.\notag
\end{align}
As depicted in Fig.~\ref{fig:truncated-DTMC}, the boundary states $(1, 0, N)$ and $(0, 1, N)$ absorb the impact of all discarded states.

In general, the truncated MDP might not converge to the original one as $N\rightarrow\infty$\cite{sennott1998stochastic}. Thus, we will be concerned with the performance loss due to state space truncation. We first show the structural results of the truncated MDP. 

\begin{proposition} 
    The optimal policy $\pi_N^*$ for the truncated problem has a switching structure. 
\end{proposition}
\begin{IEEEproof}
    The conditions in the proof of Proposition~\ref{proposition:existence-of-lambda-policy} trivially hold for the truncated finite-state MDP; hence, an optimal policy exists. The switching structure then follows by applying the same proof as Theorem~\ref{theorem:structure-lambda-optimal}.
\end{IEEEproof}

\begin{figure}
    \centering
    \includegraphics[width=\linewidth]{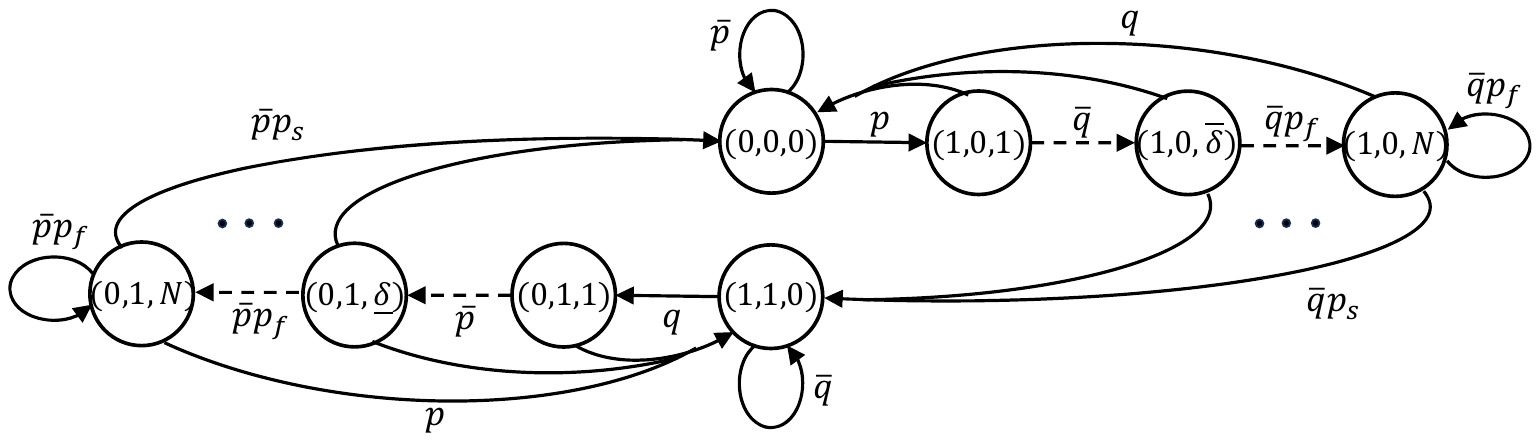}
    \caption{State evolution of the truncated MC2 under the switching policy.}
    \label{fig:truncated-DTMC}
\end{figure}

The next result shows that for any switching policy $\pi$, the average cost of the original problem, $\mathcal{L}(\pi)$, differs from that of the truncated problem, $\mathcal{L}(\pi, N)$, by a constant term $\sigma(\pi, N)$, which diminishes exponentially with $N$.

\begin{proposition}\label{proposition:truncated-MDP-convergence}
    For any a switching policy $\pi$ with thresholds $\bar{\delta},\ubar{\delta}\geq 1$ and $N>\max\{\bar{\delta},\ubar{\delta}\}$, the stationary distribution of the truncated chain, $\nu_{i,j,\delta}(N)$, is given by
    \begin{subequations}
    \begin{align}
        &\nu_{i,j,\delta}(N) = \nu_{i,j,\delta}, \,\forall (i,j,\delta)\in\mathcal{S}_N^{\rm{in}},\\
            &\nu_{1,0,N}(N) = \frac{\bar{q}p_f}{1-\bar{q}p_f} \nu_{1,0,N-1},\\
            &\nu_{0,1,N}(N) = \frac{\bar{p}p_f}{1-\bar{p}p_f}\nu_{0,1,N-1}.
    \end{align}\label{eq:stationary-distribution-truncated}
    \end{subequations}
    Moreover, the average cost of the truncated problem satisfies
    \begin{align}
        \mathcal{L}(\pi,N) = \mathcal{L}(\pi) - \sigma(\pi,N),\label{eq:average-cost-threshold-truncated}
    \end{align}
    where $\sigma(\pi,N)=
        \frac{\beta \nu_{0,0,0}p\bar{q}^N p_f^{N-\bar{\delta}+1}}{(1-\bar{q}p_f)^2} +\frac{(1-\beta)\nu_{1,1,0}q\bar{p}^N p_f^{N-\ubar{\delta}+1}}{(1-\bar{p}p_f)^2}$.
\end{proposition}
\begin{IEEEproof}
    See Appendix~\ref{proof:truncated-MDP-convergence}.
\end{IEEEproof}

Based on the above results, we show that the truncated MDP converges to the original one as the truncation is alleviated. 

\begin{theorem}\label{theorem:truncated-MDP-optimality}
    Let $\mathcal{L}^*(N)$ be the minimal cost of the truncated MDP. Then,  $\lim_{N\rightarrow\infty}\mathcal{L}^*(N)= \mathcal{L}^*$.  
\end{theorem}
\begin{IEEEproof}
    Let $\pi^*$ and $\pi^*_N$ denote the optimal policies for the original and the truncated MDP. From~\eqref{eq:average-cost-threshold-truncated}, we have
    \begin{align}
        \mathcal{L}(\pi_N^*,N)
        + \sigma(\pi^*,N) &= \mathcal{L}^*(N)
        + \sigma(\pi^*,N) \notag\\
        &\overset{(a)}{\leq} \mathcal{L}(\pi^*,N)+ \sigma(\pi^*,N) = \mathcal{L}^* \notag\\
        &\overset{(b)}{\leq} \mathcal{L}(\pi^*_N,N)+ \sigma(\pi^*_N,N)\notag\\
        & = \mathcal{L}^*(N) + \sigma(\pi^*_N,N),\notag
    \end{align}
    where $(a)$ follows because $\pi^*_N$ solves the truncated MDP, and $(b)$ follows because $\pi^*$ solves the original MDP. It gives
    \begin{align}
        \mathcal{L}^*(N)
        + \sigma(\pi^*,N) \leq
        \mathcal{L}^*
        \leq \mathcal{L}^*(N)+ \sigma(\pi^*_N,N).\notag
    \end{align}
    Tanking a limit as $N\rightarrow\infty$ yields the result.
\end{IEEEproof}

\begin{figure*}[ht]
    \begin{subfigure}{0.325\linewidth}
        \centering
        \includegraphics[width=\linewidth]{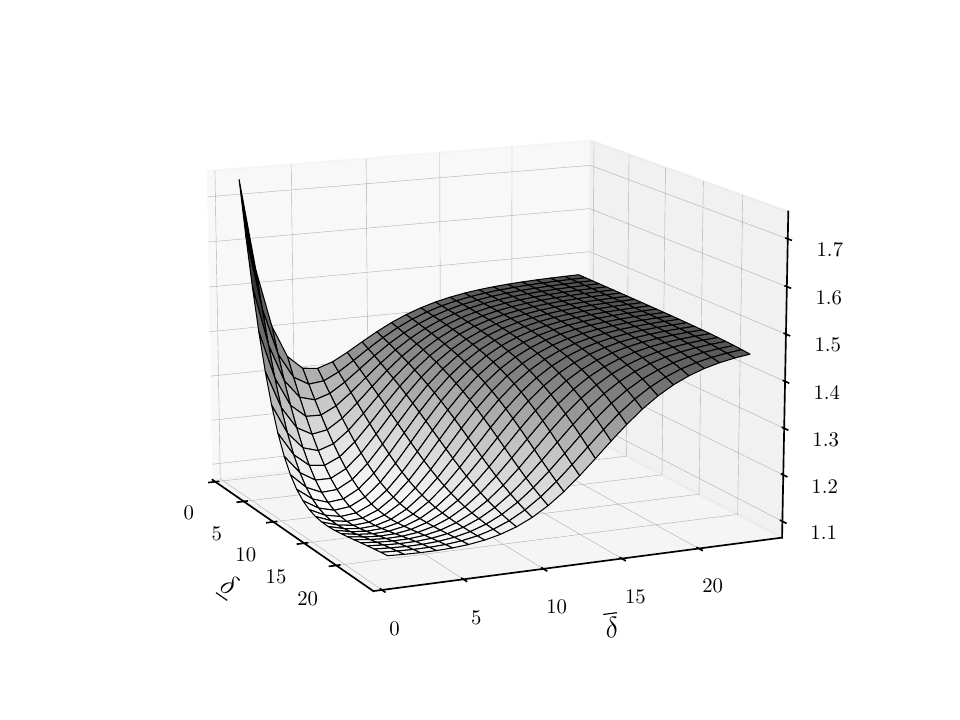}
        \caption{$p=0.2, q=0.3, \beta=0.8$.}
        \label{fig:policy-search-slow}
    \end{subfigure}
    \begin{subfigure}{0.325\linewidth}
        \centering
        \includegraphics[width=\linewidth]{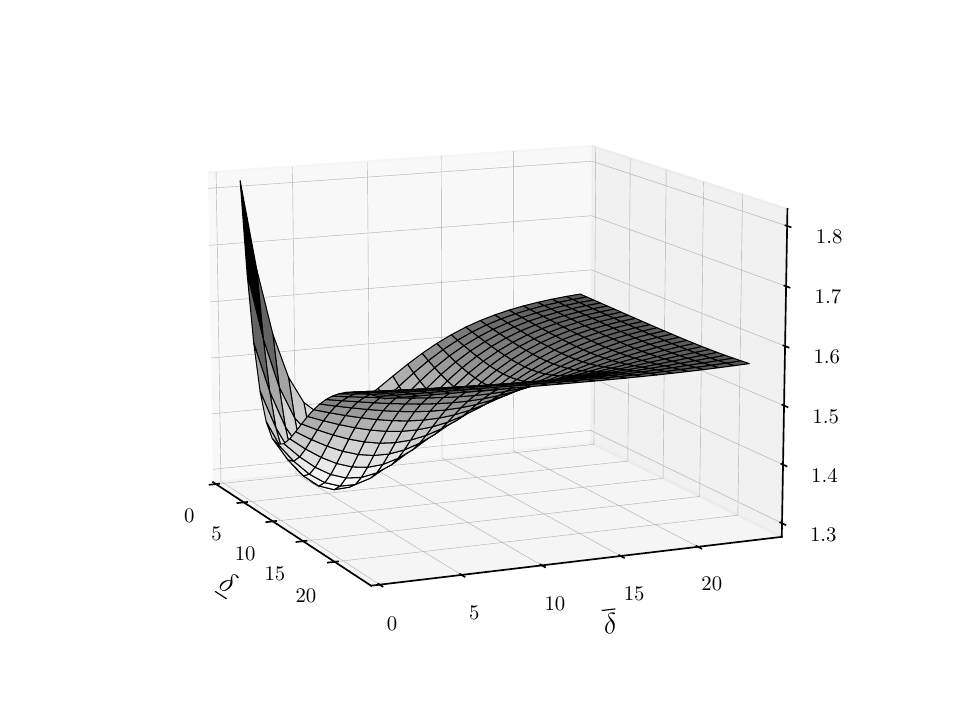}
        \caption{$p=0.25, q=0.25, \beta=0.5$.}
        \label{fig:policy-search-equal}
    \end{subfigure} 
    \begin{subfigure}{0.325\linewidth}
        \centering
        \includegraphics[width=\linewidth]{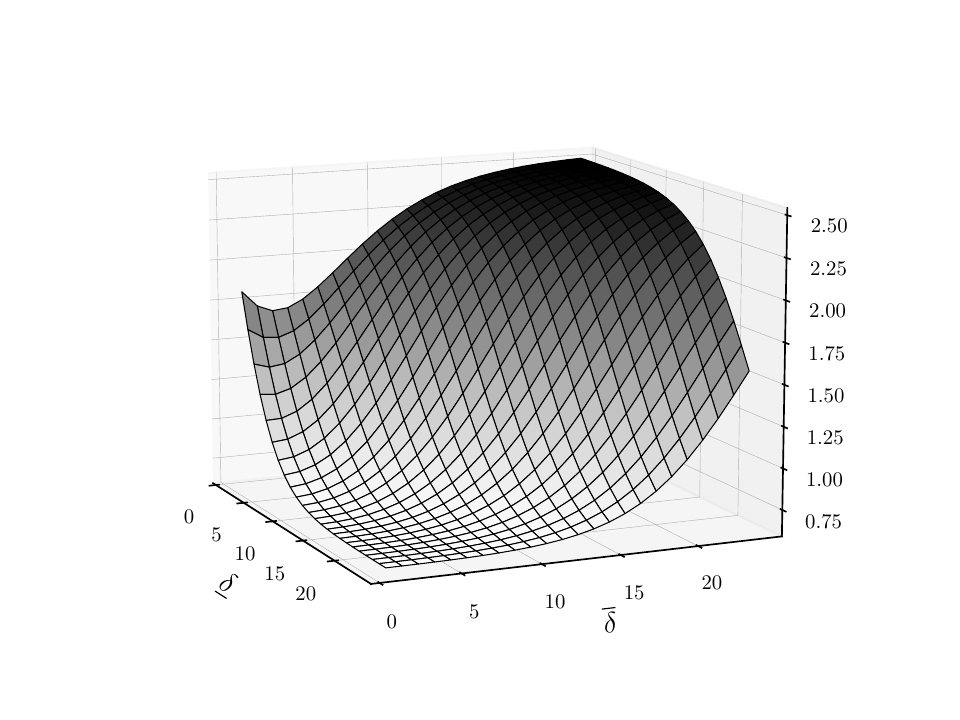}
        \caption{$p=0.25, q=0.25, \beta=0.8$.}
        \label{fig:policy-search-4}
    \end{subfigure} 
    \caption{Estimation cost as a function of the thresholds for $\lambda = 8$ and $p_s = 0.9$. The optimal thresholds $(\bar{\delta}^*, \ubar{\delta}^*)$ are: (a) $(3, 13)$, (b) $(5, 5)$, and (c) $(1, 24)$.}
    \label{fig:policy-search}
\end{figure*}

The optimal thresholds for the truncated MDP can be obtained by finding the minima of the cost function $\mathcal{L}(\pi, N)$ in the policy domain $\{\pi = (\bar{\delta}, \ubar{\delta}) : 0 \leq \bar{\delta}, \ubar{\delta} \leq N\}$. One may also apply the relative value iteration (RVI)~\cite{luo2024semantic} in the state domain $\mathcal{S}_N$; however, at the cost of high complexity. Since the impact of state-space truncation vanishes for sufficiently large $N$, this entails a trade-off between optimality and complexity.

Based on the following result, we propose an efficient search algorithm, detailed in Algorithm~\ref{alg:full}.

\begin{proposition}\label{theorem:slice-method}
    For any given AoFA threshold $\ubar{\delta}\geq 1$, there exists some finite threshold $x^*_{\ubar{\delta}}\geq 1$ such that the average cost function $\mathcal{L}(x, \ubar{\delta})$ is increasing for the all $x\geq x^*_{\ubar{\delta}}$. Similarly, for any given AoMA threshold $\bar{\delta}\geq 1$, there exists some $y^*_{\bar{\delta}} \geq 1$ such that $\mathcal{L}(\bar{\delta}, y)$ is an increasing function for all $y\geq y^*_{\bar{\delta}}$.
\end{proposition}
\begin{IEEEproof}
    See Appendix~\ref{proof:slice-method} for derivations of $x^*_{\ubar{\delta}}$ and $y^*_{\bar{\delta}}$.
\end{IEEEproof}

Algorithm~\ref{alg:full} has a time complexity of $\mathcal{O}(N^2)$, which is significantly lower than the $\mathcal{O}(|\mathcal{S}_N|^2|\mathcal{A}|I_\text{RVI})$ complexity of RVI, where $I_\text{RVI}\gg 1$ is the number of iterations required for RVI to converge.

\begin{algorithm}[ht]
\renewcommand{\thealgocf}{1}
\DontPrintSemicolon
\SetAlgoLined
\algsetup{linenosize=\small}
Choose a sufficiently large $N$ such that $\sigma(\pi, N) < \epsilon$.\; 
Compute $x^*_{\ubar{\delta}}$ and $y^*_{\bar{\delta}}$ for all $1\leq \ubar{\delta}, \bar{\delta} \leq N$. Let $\Pi_N = \{(x^*_{\ubar{\delta}},\ubar{\delta}):1\leq \ubar{\delta}\leq N\}\cap \{(\bar{\delta},y^*_{\bar{\delta}}):1\leq \bar{\delta}\leq N\}$.\;
Compute $\mathcal{L}(\pi, N)$ for all $\pi \in \Pi_N$. Find the optimal thresholds by $(\bar{\delta}^*, \ubar{\delta}^*) = \argmin_{\pi\in \Pi_N} \mathcal{L}(\pi, N)$.
\caption{The proposed policy search algorithm.}
\label{alg:full}
\end{algorithm}

\begin{figure*}[ht]
    \begin{subfigure}{0.33\linewidth}
        \centering
        \includegraphics[width=\linewidth]{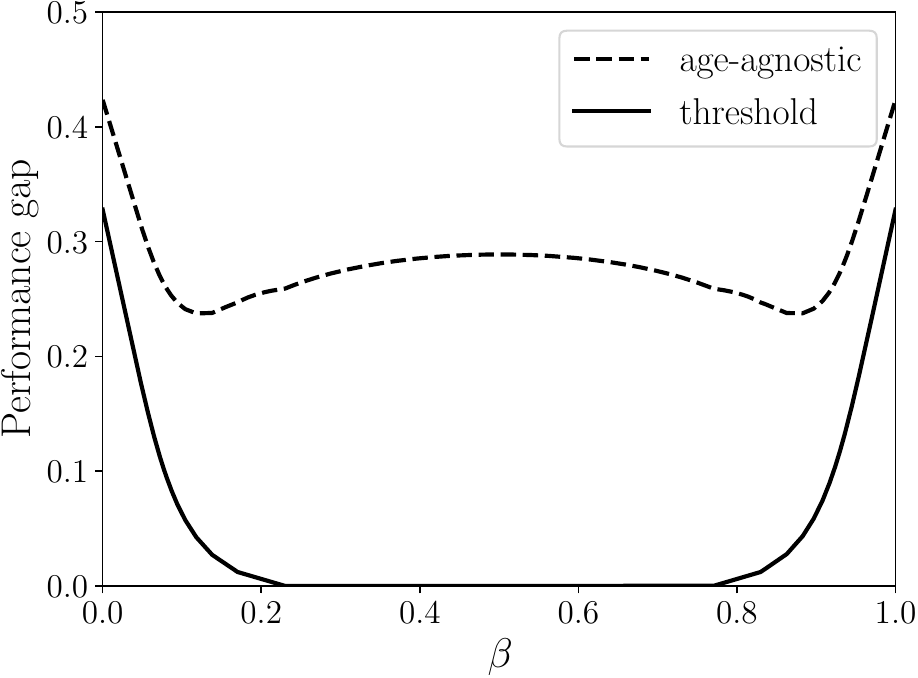}
        \caption{performance gap}
        \label{fig:performance-gap}
    \end{subfigure} 
    \begin{subfigure}{0.33\linewidth}
        \centering
        \includegraphics[width=\linewidth]{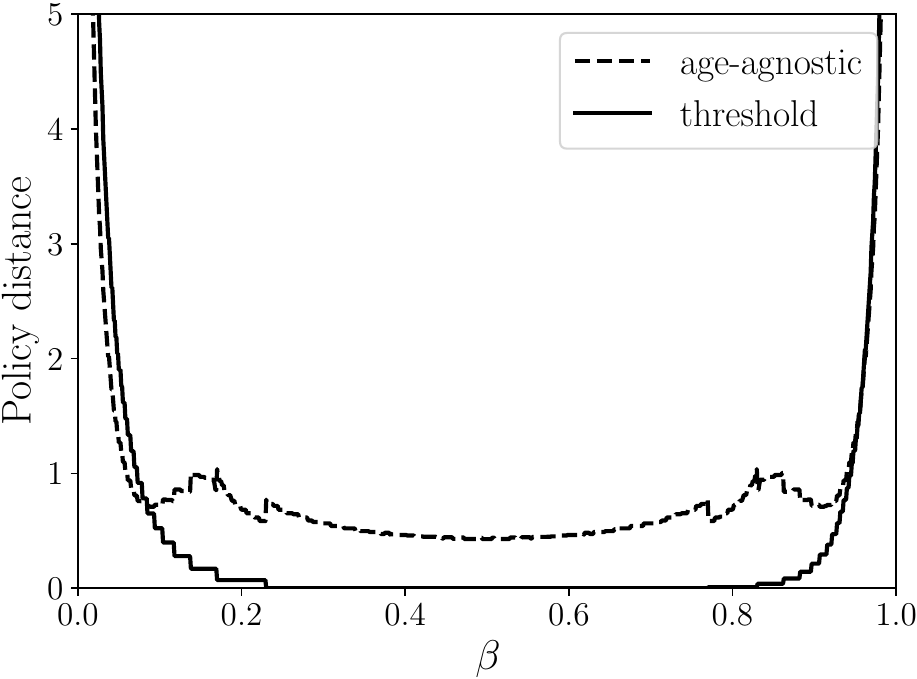}
        \caption{policy distance}
        \label{fig:performance-distance}
    \end{subfigure} 
    \begin{subfigure}{0.33\linewidth}
        \centering
        \includegraphics[width=\linewidth]{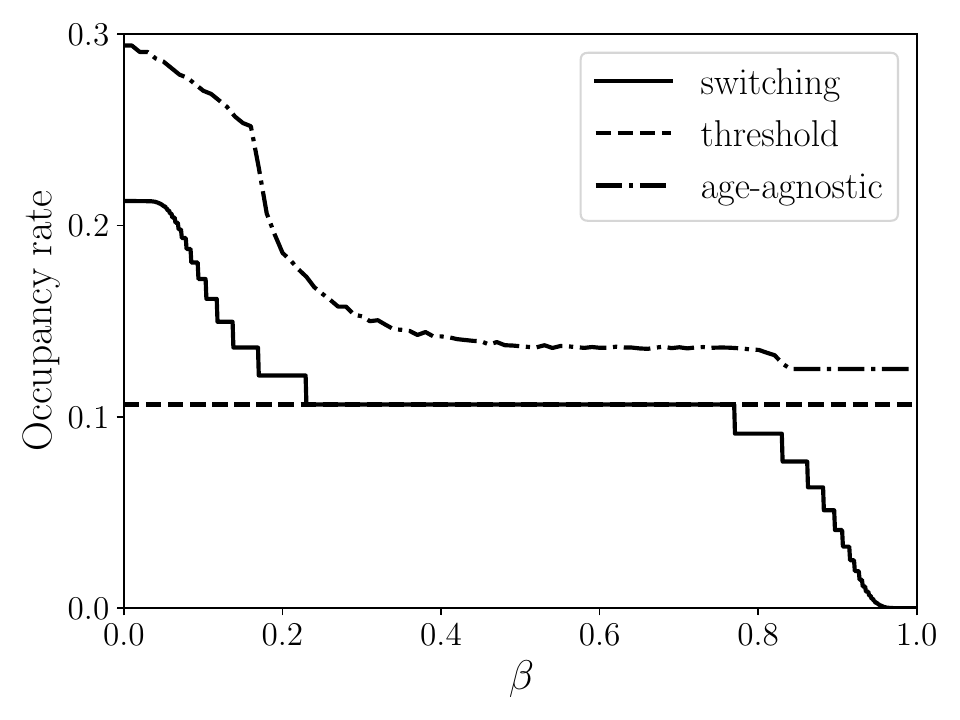}
        \caption{occupancy rate of MA errors}
        \label{fig:occupancy-rate}
    \end{subfigure}
    \caption{Performance comparison of different policies for symmetric sources with $p = q = 0.25$, $p_s = 0.9$, $\lambda = 1$, and $N = 100$.}
    \label{fig:performance-comparison}
\end{figure*}

\section{Numerical Results}\label{sec:numerical-results}
We now present numerical results on the optimal policy.

Fig.~\ref{fig:policy-search} shows the average cost $\mathcal{L}(\bar{\delta}, \ubar{\delta})$ as a function of AoMA threshold $\bar{\delta}$ and AoFA threshold $\ubar{\delta}$. Fig.~\ref{fig:policy-search-slow} demonstrates that for an asymmetric source with prioritized states, the low-cost region is where $\bar{\delta}$ is small and $\ubar{\delta}$ is large. This implies that the alarm state is significantly more important, thus forcing the sensor to transmit more frequently in MA errors while filtering out transient FA errors. Fig.~\ref{fig:policy-search-equal} illustrates a special case where the source is symmetric and the states are equally important. Here, the sensor has no preference for different types of errors; consequently, the cost surface is symmetric, and the optimal policy has identical thresholds. Fig.~\ref{fig:policy-search-4} considers a scenario in which the source is symmetric, but the alarm state has greater importance.  In this case, the threshold policy performs poorly, and the optimal strategy is to transmit in all MA errors and in sustained FA errors that exceed a large age threshold.

\begin{figure*}[ht]
    \begin{subfigure}{0.33\linewidth}
        \centering
        \includegraphics[width=\linewidth]{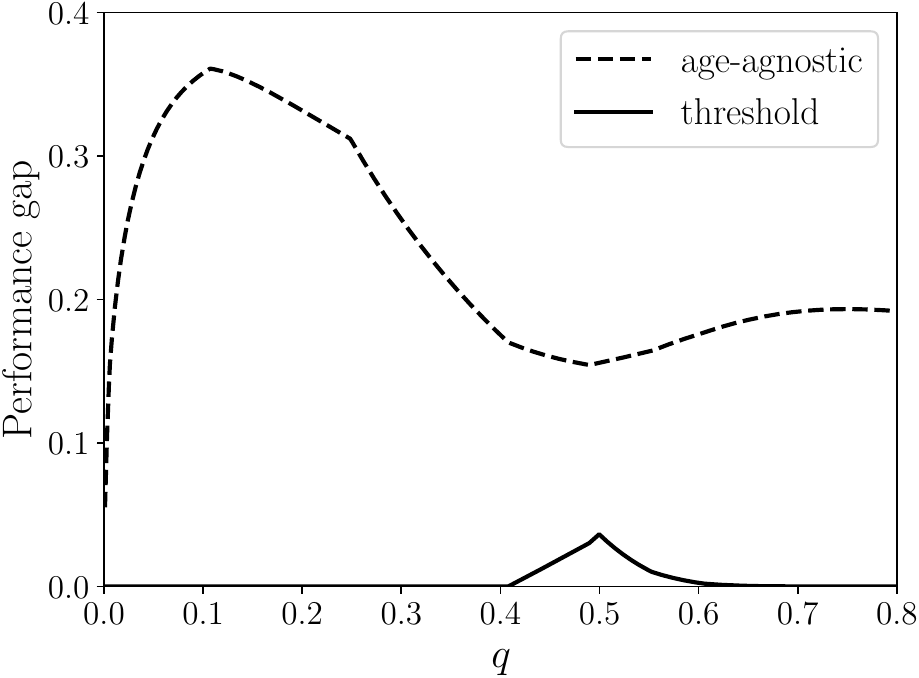}
        \caption{performance gap ($p =0.2, \lambda=1$)}
        \label{fig:performance-gap-q}
    \end{subfigure} 
    \begin{subfigure}{0.33\linewidth}
        \centering
        \includegraphics[width=\linewidth]{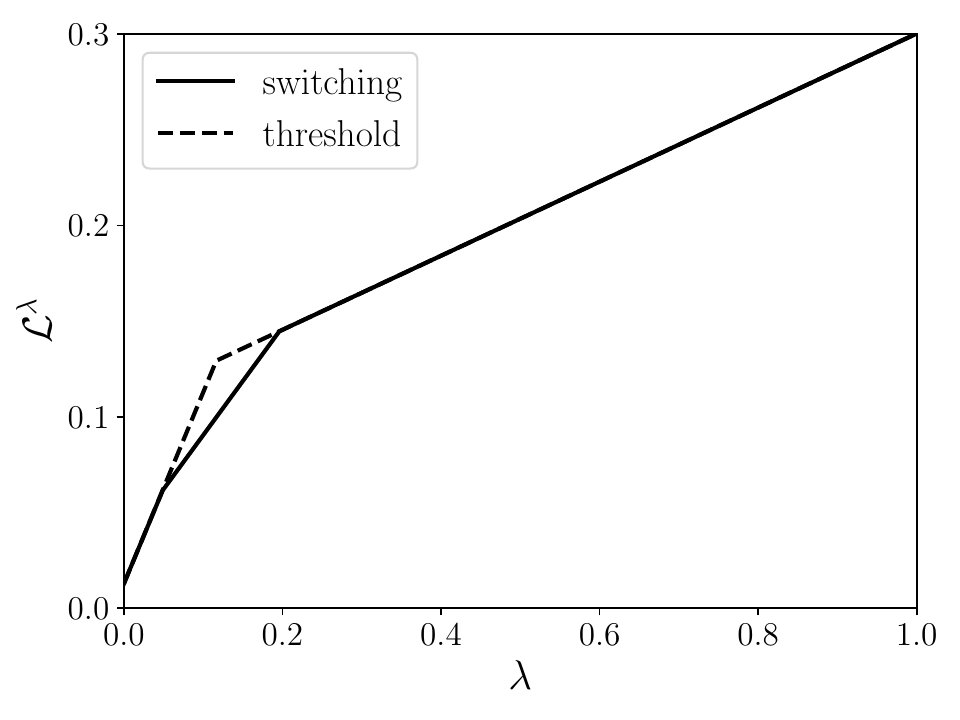}
        \caption{average cost ($p =0.2, q = 0.25$)}
        \label{fig:average-cost-lam}
    \end{subfigure} 
    \begin{subfigure}{0.33\linewidth}
        \centering
        \includegraphics[width=\linewidth]{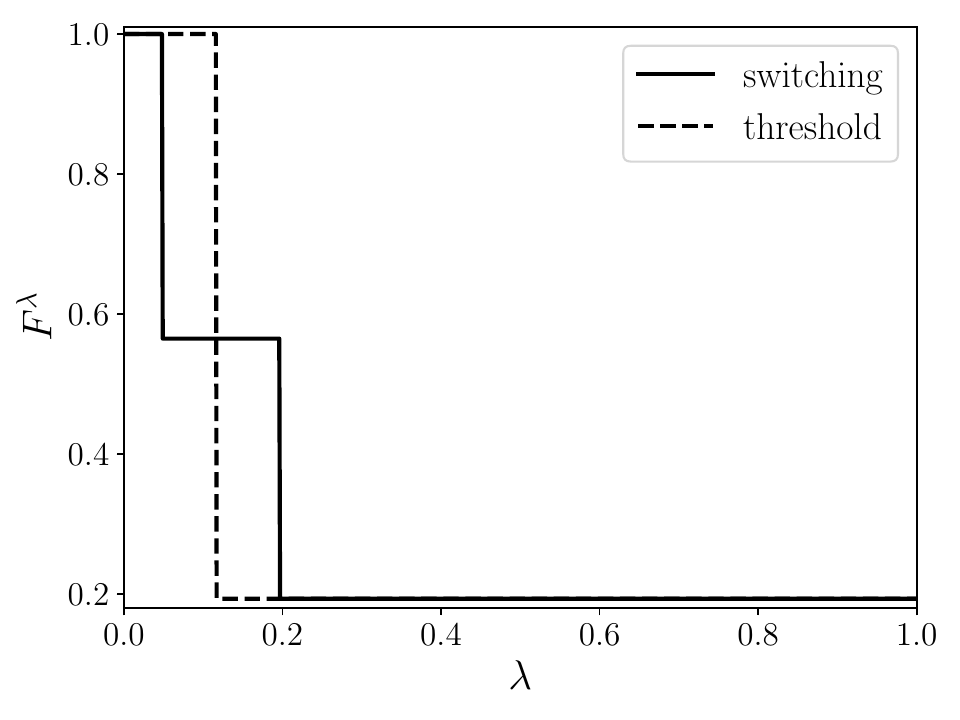}
        \caption{transmission frequency ($p =0.2, q = 0.25$)}
        \label{fig:transfrequency-lam}
    \end{subfigure} 
    \caption{Performance comparison of different policies for non-prioritized sources with $\beta=0.5$, $p_s = 0.9$, and $N = 100$.}
    \label{fig:performance-comparison-2}
\end{figure*}

Fig.~\ref{fig:performance-comparison} compares the performance achieved by the optimal switching, threshold, and age-agnostic policies for a symmetric source. The latter two policies rely on insufficient statistics: the threshold policy depends solely on the age processes and treats all source states equally, whereas the age-agnostic policy accounts for state importance but disregards the lasting impact of consecutive errors. Fig.~\ref{fig:performance-gap} and Fig.~\ref{fig:performance-distance} illustrate the performance losses incurred by these two suboptimal policies. The threshold policy achieves comparable performance and exhibits similar behavior when $\beta$ is close to $0.5$. In contrast, its performance deteriorates rapidly as the source becomes more prioritized. The age-agnostic policy, however, shows a large performance gap and markedly different behavior even when the source is non-prioritized.

Fig.~\ref{fig:occupancy-rate} compares the occupancy rate of MA errors under different policies. A higher occupancy rate indicates a larger proportion of time spent in MA errors. The threshold policy remains unchanged for different MA significance levels $\beta$. In contrast, the age-agnostic and switching policies consistently adapt their behaviors according to state importance: as $\beta$ increases, the sensor transmits more frequently during MA errors, thereby reducing the occupancy rate. However, the age-agnostic policy often transmits at inappropriate times: it tends to overreact in FA errors when $\beta$ is small and becomes overly conservative in MA errors when $\beta$ is large. CAE-optimal policy (not plotted) triggers transmissions only in response to urgent errors: for $\beta>0.5$ (or $\beta<0.5$), transmissions occur exclusively during MA (FA) errors, which can result in significant lasting impacts in the less important error.

Fig.~\ref{fig:performance-comparison-2} compares the performance of different policies for a non-prioritized source. Fig.~\ref{fig:performance-gap-q} shows that the threshold policy achieves near-optimal performance except within a narrow range of $q$ values for a given $p$, whereas the age-agnostic policy incurs a significant performance gap. Figs.~\ref{fig:average-cost-lam} and \ref{fig:transfrequency-lam} further illustrate the structure of the optimal policy: (i) always transmit when $\lambda < 0.05$; (ii) transmit in the synced state $(0,0,0)$ when $\lambda \in [0.05, 0.20)$; and (iii) transmit only when the AoMA or AoFA exceeds a positive threshold when $\lambda > 0.20$. By contrast, the threshold policy offers only two options: continuous transmission when $\lambda < 0.12$, or completely ignoring the synced states otherwise.

\section{Conclusion}\label{sec:conclusion}
This paper studied semantics-aware remote estimation of Markov sources. To measure data significance, we assigned different costs and age variables (i.e., AoMA and AoFA) to different types of estimation errors. The main results were: Theorem~\ref{theorem:structure-lambda-optimal} proved the existence of an optimal switching policy. Theorem~\ref{theorem:average-cost-threshold} provided analytical results for switching policies. Theorem~\ref{theorem:AoII-threshold} justified the intuition that the optimal policy has identical thresholds when the source is symmetric and the states are equally important. Theorem~\ref{theorem:truncated-MDP-optimality} demonstrated the asymptotic optimality of the truncated MDP. Furthermore, an efficient search algorithm was proposed to compute the optimal policy. Numerical results illustrate the effectiveness of incorporating data significance into the estimation process.

\appendices

\renewcommand{\theequation}{A\arabic{equation}}
\setcounter{equation}{0} 

\section{Proof of Proposition~\ref{proposition:sa-policy}}\label{proof:sa-policy}
\textit{Part (i):} As depicted in Fig.~\ref{fig:state_DTMC}, every state in MC1 has a positive self-transition probability (i.e., aperiodicity) and is reachable from any other state (i.e., recurrence). Therefore, MC1 is ergodic and admits a unique stationary distribution~\cite[Ch.~4]{gallager1997discrete} that satisfies the following balance equations: 
\begin{align}
    \begin{cases}
        \nu_{0,0} = \bar{p}\nu_{0,0} + \bar{p}f^0_\alpha\nu_{0, 1} + q\nu_{1,0} + qf^1_\alpha\nu_{1,1},\\
    \nu_{0,1} = \bar{p}\bar{f}^0_\alpha \nu_{0, 1} + q\bar{f}^1_\alpha \nu_{1,1},\\
    \nu_{1,0} = p\bar{f}^0_\alpha \nu_{0,0} + \bar{q}\bar{f}^1_\alpha \nu_{1,0},\\
    \nu_{1,1} = pf^0_\alpha\nu_{0,0} + p\nu_{0, 1} + \bar{q}f^1_\alpha\nu_{1,0} + \bar{q}\nu_{1,1},\\
    \textstyle\sum\nolimits_{i,j}\nu_{i,j} = 1.
    \end{cases}\label{eq:prove-A-1}
\end{align}
The results in~\eqref{eq:stationary-distribution-sa} are the solution to the linear system~\eqref{eq:prove-A-1}.

\textit{Part (ii):} As shown in Fig.~\ref{fig:joint_status_DTMC}, MC2 has no isolated groups (i.e., transitions exist for all pairs of states), and $(0,0,0)$ is a positive recurrent state. Therefore, MC2 is irreducible and all states are positive recurrent\cite[Ch.~5]{gallager1997discrete}. The probabilities in~\eqref{eq:sa-stationary} are obtained by multiplying the transition probabilities along the paths from $(0,0,0)$ to $(1,0,k)$ and from $(1,1,0)$ to $(0,1,k)$. Since MC2 is stationary, the average cost incurred by the age-agnostic policy $\pi$ equals the expected total cost of being in the states. Hence, the average cost is computed as
\begin{align}
    &C(\pi) = \beta\sum_{k=1}^{\infty}k\nu_{1,0,k}
    + (1-\beta)\sum_{k=1}^{\infty}k\nu_{0,1,k}\notag\\
    &= \beta\sum_{k=1}^{\infty}k p\bar{f}^0_\alpha(\bar{q}\bar{f}^1_\alpha)^{k-1}\nu_{0,0}+(1-\beta)\sum_{k=1}^{\infty}k q\bar{f}^1_\alpha(\bar{p}\bar{f}^0_\alpha)^{k-1}\nu_{1,1}.\notag
\end{align}
Given that $0<\bar{f}^0_\alpha, \bar{f}^1_\alpha<1$, the above summations converge to a finite constant, i.e.,
\begin{align}
    C(\pi) = \frac{\beta p\bar{f}^0_\alpha\nu_{0,0}}{(1-\bar{q}\bar{f}^1_\alpha)^2} + \frac{(1-\beta)q\bar{f}^1_\alpha\nu_{1,1}}{(1-\bar{p}\bar{f}^0_\alpha)^2} < \infty.
\end{align}
The transmission frequency is given by
\begin{align}
    F(\pi) = \sum_{i}f_i \nu_i = \frac{qf_0+pf_1}{p+q}.
\end{align}
Then, the total average cost is obtained as
\begin{align}
    \mathcal{L}(\pi) &= C(\pi) + \lambda F(\pi)\notag\\
    &= \frac{\beta p\bar{f}^0_\alpha\nu_{0,0}}{(1-\bar{q}\bar{f}^1_\alpha)^2} + \frac{(1-\beta)q\bar{f}^1_\alpha\nu_{1,1}}{(1-\bar{p}\bar{f}^0_\alpha)^2} + \lambda \frac{qf_0+pf_1}{p+q},\label{eq:prove-A-2}
\end{align}
which yields the result in~\eqref{eq:sa-cost}.

\textit{Part (iii):} Let $z = (0,0,0)$. For brevity, we only verify MA errors, i.e., $i = (1, 0, \delta), \forall\delta\geq 1$. The derivation for other states follows analogously. There are three possible routes from state $i$ to $z$, as illustrated in Fig.~\ref{fig:possible-routes}. For route (a), the expected first passage time from state $i$ to $z$ is
\begin{align}
    \bar{T}_{i,z} = \sum_{t = 1}^{\infty}t  ({\bar{q}\bar{f}^1_\alpha})^{t-1}  q^1 = \frac{q}{(1-\bar{q}\bar{f}^1_\alpha)^2}<\infty.\label{eq:prove-A-3}
\end{align}
For route (b), the system may stay in $(1, 1)$ and $(0,1)$ for some time or alternate between them. Let 
\begin{align*}
    \kappa = \max\{\bar{q}\bar{f}^1_\alpha, \bar{q}f^1_\alpha, \bar{q}, q\bar{f}^0_\alpha, p, \bar{p}\bar{f}^0_\alpha\} < 1
\end{align*}
be the largest transition probability along route (a). The expected first passage time satisfies
\begin{align}
    \bar{T}_{i,z} < \sum_{t=1}^{\infty}t
     \kappa^{t-1} (\bar{p}f^0_\alpha)^1 = \frac{\bar{p}f^0_\alpha}{(1-\kappa)^2} <\infty. \label{eq:prove-A-4}
\end{align}
Similarly, one can show that the first passage time for route (c) is also finite. This proves part (iii).

\textit{Part (iv):} The first passage cost from $i$ to $z$ is the accumulated expected cost along the routes. For route (a), we have
\begin{align}
    \bar{M}_{i, z} &= \sum_{t=0}^{\infty}(\bar{q}\bar{f}^1_\alpha)^{t}\beta(\delta+t) =\beta\delta\sum_{t=0}^{\infty}(\bar{q}\bar{f}^1_\alpha)^{t} + \beta \sum_{t=0}^{\infty}t(\bar{q}\bar{f}^1_\alpha)^{t} \notag\\
    &= \frac{\beta \delta}{1-\bar{q}\bar{f}^1_\alpha}+\frac{\beta \bar{q}\bar{f}^1_\alpha}{(1-\bar{q}\bar{f}^1_\alpha)^2}<\infty.\label{eq:prove-A-5}
\end{align}
Similarly, one can show that the costs associated with routes (b) and (c) are finite. This establishes part (iv).

\begin{figure}[t]
    \centering
    \scalebox{1}{\begin{tikzpicture}[node distance=1.6cm]
\tikzset{
  state/.style={circle, draw, minimum size=0.25cm, inner sep=0pt, font=\scriptsize, line width=0.25mm},
  loop above/.style={->, loop, looseness=4, out=120, in=60, line width=0.2mm},
  every edge/.style={draw, ->, line width=0.2mm}
}

\node (labela) at (-0.8,0) {(a)};
\node[state] (a10) at (0,0) {$(1,0)$};
\node[state] (a00) [right of=a10] {$(0,0)$};
\path (a10) edge[loop above] (a10)
      (a10) edge (a00);

\node (labelb) at (-0.8,-1.2) {(b)};
\node[state] (b10) at (0,-1.2) {$(1,0)$};
\node[state] (b11) [right of=b10] {$(1,1)$};
\node[state] (b01) [right of=b11] {$(0,1)$};
\node[state] (b00) [right of=b01] {$(0,0)$};
\path (b10) edge[loop above] (b10)
      (b11) edge[loop above] (b11)
      (b01) edge[loop above] (b01)
      (b10) edge[->] (b11)
      (b11) edge[<->] (b01)
      (b01) edge[->] (b00);

\node (labelc) at (-0.8,-2.4) {(c)};
\node[state] (c10) at (0,-2.4) {$(1,0)$};
\node[state] (c11) [right of=c10] {$(1,1)$};
\node[state] (c01) [right of=c11] {$(0,1)$};
\node[state] (c11b) [right of=c01] {$(1,1)$};
\node[state] (c00) [right of=c11b] {$(0,0)$};
\path (c10) edge (c11)
      (c11) edge[<->] (c01)
      (c01) edge[<->] (c11b)
      (c11b) edge (c00)
      (c10) edge[loop above] (c10)
      (c11) edge[loop above] (c11)
      (c01) edge[loop above] (c01)
      (c11b) edge[loop above] (c11b);
\end{tikzpicture}}
    \caption{Possible routes from state $(1,0,\delta)$ to $(0,0,0)$.}
    \label{fig:possible-routes}
\end{figure}

\renewcommand{\theequation}{B\arabic{equation}}
\setcounter{equation}{0} 

\section{Proof of Theorem~\ref{theorem:structure-lambda-optimal}}\label{poof:structure-lambda-optimal}
It suffices to show that the optimal policy is monotonically nondecreasing in each age process. In the following proof, we show the monotonicity of the optimal policy with respect to the AoMA. We begin by defining the submodularity of multivariable functions, which is fundamental to our analysis.
\begin{definition}
    Let $\mathcal{X}$ and $\mathcal{Y}$ be ordered sets, and let $f(x,y)$ be a real-valued function on $\mathcal{X}\times\mathcal{Y}$. The function $f(x,y)$ is called submodular if, for all $x_{-}\leq x_{+}$ and $y_{-}\leq y_{+}$,
    \begin{align}
        f(x_{+}, y_{+}) + f(x_{-}, y_{-}) \leq f(x_{+}, y_{-}) + f(x_{-}, y_{+}).\notag
    \end{align}
    If the inequality is reversed, $f(x,y)$ is called supermodular.
\end{definition}
The key result behind submodularity is the following (see, e.g.,~\cite[Lemma~4.7.1]{puterman1994markov}).
\begin{lemma}\label{lemma:submodularity}
    Let $f(x,y)$ be a submodular function on $\mathcal{X}\times\mathcal{Y}$ and assume that for all $x$, $\min_{y\in\mathcal{Y}}f(x,y)$ exists. Then $g(x) = \max\{y^\prime \in \argmin_{y\in\mathcal{Y}}f(x,y)\}$ is nondecreasing in $x$.
\end{lemma}
Recall that the optimal policy is obtained by 
\begin{align}
    \pi^*(i) = \argmin_{a}\{\ell(i, a) + \sum\nolimits_{j}P_{i,j}(a)h(j)\}.\notag
\end{align}
The proof is reduced to verifying the following conditions (see, e.g.,~\cite[Theorem~8.11.3]{puterman1994markov}):
\begin{itemize}
    \item[C1:] $\ell(s,a)$ is nondecreasing in $s\in\mathcal{S}_\text{MA}$ for all $a\in\mathcal{A}$;
    \item[C2:] $q(k|s,a):=\sum_{j=k}^{\infty}\Pr[j|s,a]$ is nondecreasing in $s\in\mathcal{S}_\text{MA}$ for all $k$ and $a\in\mathcal{A}$; 
    \item[C3:] $\ell(s,a)$ is submodular on $\mathcal{S}_\text{MA}\times\mathcal{A}$;
    \item[C4:] and $q(k|s,a)$ is submodular on $\mathcal{S}_\text{MA}\times\mathcal{A}$ for all $k$.
\end{itemize}

We first establish the monotonicity of the functions $\ell(s, a)$ and $q(s^\prime|s, a)$. For any MA error $s \in \mathcal{S}_\text{MA}$, the possibilities of the next states are
\begin{align}
    \Pr[s^\prime|s,a] &= \begin{cases}
    q , &s^\prime=(0, 0, 0), a=0/1,\\
    \bar{q} p_s, &s^\prime=(1, 1, 0), a=1,\\
    \bar{q}p_f, &s^\prime=(1, 0, \delta+1), a=1,\\
    \bar{q}, &s^\prime=(1, 0, \delta+1), a=0.
    \end{cases}\notag
\end{align}
From~\eqref{eq:expected-age-cost}, the cost of taking action $a$ in state $s$ is
\begin{align}
    \ell(s,a) = \begin{cases}
        \beta (\delta+1)\bar{q}p_f + \lambda, &a = 1,\\
        \beta (\delta+1)\bar{q}, &a = 0.
    \end{cases}\label{eq:prov-mono-a}
\end{align}
It is clear that $\ell(s,a)$ is increasing in $s\in\mathcal{S}_\text{MA}$.

For the synced state $s=(0,0,0)$, we may write
\begin{align}
    \Pr[s^\prime|s,a] &= \begin{cases}
    \Bar{p} , &s^\prime=(0, 0, 0), a=0/1,\\
     p p_s, &s^\prime=(1, 1, 0), a=1,\\
    p p_f, &s^\prime=(1, 0, 1), a=1,\\
    p, &s^\prime=(1, 0, 1), a=0,
    \end{cases}\notag
\end{align}
and 
\begin{align}
    \ell(s,a) = \begin{cases}
        \beta pp_f + \lambda, &a = 1,\\
        \beta p, &a = 0.
    \end{cases}\label{eq:prov-mono-b}
\end{align}
From \eqref{eq:prov-mono-a}-\eqref{eq:prov-mono-b}, we conclude that, when $p\leq 2\bar{q}$, $\ell(s,a)$ is monotonically nondecreasing for all $s\in\{(0,0,0)\}\cup\mathcal{S}_\text{MA}$.

The quantity $q(k|s,a)$ for $s = (1,0, \delta)$ is given by
\begin{align}
    q(k|s,a=1) &= \begin{cases}
        1, &k\leq (\cdot, 0), \\
        \bar{q}p_f, &(\cdot, 0)<k\leq (\cdot, \delta+1), \\
        0, &k > (\cdot, \delta+1),
    \end{cases}\label{eq:prove-q-0}\\
    q(k|s,a=0) &= \begin{cases}
        1, &k\leq (\cdot, 0), \\
        \bar{q}, &(\cdot, 0)<k\leq (\cdot, \delta+1), \\
        0, &k > (\cdot, \delta+1).
    \end{cases}\label{eq:prove-q-1}
\end{align}
Hence, $q(k|s,a)$ is nondecreasing for all $s\in\mathcal{S}_\text{MA}$. For the synced state $s=(0,0,0)$, we have
\begin{align}
    q(k|s,a=1) &= \begin{cases}
        1, &k\leq (\cdot, 0), \\
        pp_f, &(\cdot, 0)<k\leq (\cdot, 1), \\
        0, &k > (\cdot, 1),
    \end{cases}\label{eq:prove-q-2}\\
    q(k|s,a=0) &= \begin{cases}
        1, &k\leq (\cdot, 0), \\
        p, &(\cdot, 0)<k\leq (\cdot, 1), \\
        0, &k > (\cdot, 1).
    \end{cases}\label{eq:prove-q-3}
\end{align}
From \eqref{eq:prove-q-0}-\eqref{eq:prove-q-3}, it follows that $q(k|s,a)$ is nondecreasing for all $s\in\{(0,0,0)\}\cup\mathcal{S}_\text{MA}$ if $p\leq \bar{q}$. Therefore, we conclude that $\ell(s, a)$ and $q(k|s,a)$ are nondecreasing in $\mathcal{S}_\text{MA}$, and that monotonicity extends to $\{(0,0,0)\}\cup\mathcal{S}_\text{MA}$ if $p\leq \bar{q}$.

We now establish the submodularity of $\ell(s, a)$ and $q(k|s, a)$. Let $s_{-}\leq s_{+}$, $s_{-}, s_{+}\in\mathcal{S}_\text{MA}$, and $a_{-} = 0, a_{+} = 1$. Then
\begin{align}
    \ell(s_{+},a_{+}) &+ \ell(s_{-},a_{-}) - \big(\ell(s_{+},a_{-}) + \ell(s_{-},a_{+})\big) \notag\\
    &= \beta (\delta_{+}+1)\bar{q}p_f + \lambda + \beta(\delta_{-}+1)\bar{q}\notag\\
    &\qquad -\big(\beta (\delta_{+}+1)\bar{q} + \beta(\delta_{-}+1)\bar{q}p_f + \lambda\big)\notag\\
    &= \bar{q}p_s\beta (\delta_{-}-\delta_{+})\leq 0,\label{eq:prove-l-0}
\end{align}
and
\begin{align}
    q(k|s_{+},a_{+}) &+ q(k|s_{-},a_{-}) \notag\\
    &= \begin{cases}
        2, &k\leq (\cdot, 0), \\
        \bar{q}(1+p_f), &(\cdot, 0)<k\leq (\cdot, \delta_{-}), \\
        \bar{q}p_f, &(\cdot, \delta_{-})<k\leq (\cdot, \delta_{+}), \\
        0, &k > (\cdot, \delta_{+}).
    \end{cases}\\
    q(k|s_{+},a_{-}) &+ q(k|s_{-},a_{+}) \notag\\
    &= \begin{cases}
        2, &k\leq (\cdot, 0), \\
        \bar{q}(1+p_f), &(\cdot, 0)<k\leq (\cdot, \delta_{-}), \\
        \bar{q}, &(\cdot, \delta_{-})<k\leq (\cdot, \delta_{+}), \\
        0, &k > (\cdot, \delta_{+}).\label{eq:prove-l-1}
    \end{cases}
\end{align}
From \eqref{eq:prove-l-0}-\eqref{eq:prove-l-1}, we conclude that $\ell(s,a)$ and $q(k|s,a)$ are submodular for all $s\in\mathcal{S}_\text{MA}$.

For the synced state $s_{-} = (0,0,0)$, we have
\begin{align}
    \ell(s_{+},a_{+}) + \ell(s_{-},a_{-}) &- \big(\ell(s_{+},a_{-}) + \ell(s_{-},a_{+})\big) \notag\\
    &= \beta (\delta_{+}+1)\bar{q}p_f + \lambda + \beta p\notag\\
    &-\big(\beta (\delta_{+}+1)\bar{q} + \beta p p_f + \lambda\big)\notag\\
    &= p_s\beta \big(p-\bar{q}(\delta_{+}+1)\big)\notag\\
    &\leq p_s\beta (p-2\bar{q}),\label{eq:proof-last-ineq}
\end{align}
and
\begin{align}
    q(k|s_{+},a_{+}) &+ q(k|s_{-},a_{-})\notag\\
    &= \begin{cases}
        2, &k\leq (\cdot, 0), \\
        p+\bar{q}p_f, &(\cdot, 0)<k\leq (\cdot, 1), \\
        \bar{q}p_f, &(\cdot, 1)<k\leq (\cdot, \delta_{+}), \\
        0, &k > (\cdot, \delta_{+}).
    \end{cases}\\
    q(k|s_{+},a_{-}) &+ q(k|s_{-},a_{+}) \notag\\
    &= \begin{cases}
        2, &k\leq (\cdot, 0), \\
        \bar{q}+ pp_f, &(\cdot, 0)<k\leq (\cdot, 1), \\
        \bar{q}, &(\cdot, 1)<k\leq (\cdot, \delta_{+}), \\
        0, &k > (\cdot, \delta_{+}).
    \end{cases}\label{eq:proof-last-ineq2}
\end{align}
From \eqref{eq:proof-last-ineq}-\eqref{eq:proof-last-ineq2}, submodularity holds for $(0, 0, 0)$ if
\begin{align}
    p-2\bar{q}\leq 0~~\text{and}~~ p + \bar{q}p_f \leq \bar{q} + pp_f.\notag
\end{align}
Rearranging the terms gives $p\leq \bar{q}$. Therefore, we conclude that $\ell(s, a)$ and $q(k|s,a)$ are submodular in $\mathcal{S}_\text{MA}$, and that submodularity extends to $\{(0,0,0)\}\cup\mathcal{S}_\text{MA}$ if $p\leq \bar{q}$.

Similarly, we can show that monotonicity and submodularity hold for all $s\in\mathcal{S}_\text{FA}$. These results extend to $\{(1,1,0)\}\cup\mathcal{S}_\text{FA}$ if $q\leq \bar{p}$. Note that $p\leq \bar{q}$ implies $q\leq \bar{p}$. This completes the proof.

\renewcommand{\theequation}{C\arabic{equation}}
\setcounter{equation}{0} 

\section{Proof of Proposition~\ref{proposition:stationary-distribution-threshold}}\label{proof:stationary-distribution-threshold}
We first introduce two probability functions, $a(n)$ and $b(n)$, which will find application in various places. Define
\begin{align}
    a(n) = \sum_{k=n}^{\infty}\frac{\nu_{1,0,k}}{\nu_{0,0,0}},\,\, b(n) = \sum_{k=n}^{\infty}\frac{\nu_{0,1,k}}{\nu_{1,1,0}}.
\end{align}
The total stationary probability of all MA errors is
\begin{align}
\sum_{k=1}^{\infty}\nu_{1,0,k}&=\sum_{k=1}^{\bar{\delta}-1}p\bar{q}^{k-1}\nu_{0,0,0} + \sum_{k=\bar{\delta}}^{\infty}p\bar{q}^{\bar{\delta}-1}(\bar{q}p_f)^{k-\bar{\delta}}\nu_{0,0,0}\notag\\
&=\sum_{k=1}^{\bar{\delta}-1}p\bar{q}^{k-1}\nu_{0,0,0} + \sum_{k^\prime=0}^{\infty}p\bar{q}^{\bar{\delta}-1}(\bar{q}p_f)^{k^\prime}\nu_{0,0,0}\notag\\
    &=\left(\frac{p(1-\bar{q}^{\bar{\delta}-1})}{1-\bar{q}} + \frac{p\bar{q}^{\bar{\delta}-1}}{1-\bar{q}p_f}\right)\nu_{0,0,0}.\notag
\end{align}
It gives that
\begin{align}
    a(1) = \frac{p(1-\bar{q}^{\bar{\delta}-1})}{1-\bar{q}} + \frac{p\bar{q}^{\bar{\delta}-1}}{1-\bar{q}p_f}, \, a(\bar{\delta})= \frac{p\bar{q}^{\bar{\delta}-1}}{1-\bar{q}p_f}.\notag
\end{align}
Similarly, the total stationary probability of all FA errors is
\begin{align}
    \sum_{k=1}^{\infty}\nu_{0,1,k} &= \sum_{k=1}^{\ubar{\delta}-1}q\bar{p}^{k-1}\nu_{1,1,0} + \sum_{k=\ubar{\delta}}^{\infty}q\bar{p}^{\ubar{\delta}-1}(\bar{p}p_f)^{k-\ubar{\delta}}\nu_{1,1,0}\notag\\
    &= \sum_{k=1}^{\ubar{\delta}-1}q\bar{p}^{k-1}\nu_{1,1,0} + \sum_{k^\prime=0}^{\infty}q\bar{p}^{\ubar{\delta}-1}(\bar{p}p_f)^{k^\prime}\nu_{1,1,0}\notag\\
    &=\left(\frac{q(1-\bar{p}^{\ubar{\delta}-1})}{1-\bar{p}} + \frac{q\bar{p}^{\ubar{\delta}-1}}{1-\bar{p}p_f}\right)\nu_{1,1,0}.\notag
\end{align}
It gives that
\begin{align}
    b(1) = \frac{q(1-\bar{p}^{\ubar{\delta}-1})}{1-\bar{p}} + \frac{q\bar{p}^{\ubar{\delta}-1}}{1-\bar{p}p_f}, \,\, b(\ubar{\delta}) = \frac{q\bar{p}^{\ubar{\delta}-1}}{1-\bar{p}p_f}.\notag
\end{align}

As the stationary distribution (if it exists) forms a probability distribution, we have
\begin{align}
    (1+a(1))\nu_{0,0,0}+(1+b(1))\nu_{1,1,0} = 1.\label{eq:prov-total}
\end{align}
It must also satisfies the following balance equations:
\begin{align}
    \nu_{0,0,0} &= \bar{p}\nu_{0,0,0} + qa(1)\nu_{0,0,0} + \bar{p}p_s b(\ubar{\delta})\nu_{1,1,0},\notag\\
    \nu_{1,1,0} &= \bar{q}\nu_{1,1,0} + pb(1)\nu_{1,1,0} + \bar{q}p_s a(\bar{\delta})\nu_{0,0,0}.\notag
\end{align}
Rearranging the terms gives
\begin{align}
    (p-qa(1))\nu_{0,0,0} &= \bar{p}p_s b(\ubar{\delta})\nu_{1,1,0},\label{eq:prov-synced-0}\\
    (q-pb(1))\nu_{1,1,0} &= \bar{q}p_s a(\bar{\delta})\nu_{0,0,0}\label{eq:prov-synced-1}.
\end{align}
The stationary distribution exists if the linear system \eqref{eq:prov-total}-\eqref{eq:prov-synced-1} admits a unique solution. Since $\nu_{0,0,0}$ and $\nu_{1,1,0}$ are the only free variables, the proof reduces to showing that~\eqref{eq:prov-synced-0} and \eqref{eq:prov-synced-1} are equivalent. We establish this by noting that
\begin{align}
    p-qa(1) &= p - q \bigg(
    \frac{p(1-\bar{q}^{\bar{\delta}-1})}{1-\bar{q}} + \frac{p\bar{q}^{\bar{\delta}-1}}{1-\bar{q}p_f}
    \bigg)=\bar{q}p_s a(\bar{\delta}),\notag\\
    q-pb(1) &= q - b \bigg(
    \frac{q(1-\bar{p}^{\ubar{\delta}-1})}{1-\bar{p}} + \frac{q\bar{p}^{\ubar{\delta}-1}}{1-\bar{p}p_f}\bigg)=\bar{p}p_s b(\ubar{\delta}).\notag
\end{align}
Substituting the equality $p-qa(1) = \bar{p}p_s b(\ubar{\delta})$ into \eqref{eq:prov-synced-0} yields 
\begin{align}
\bar{q}a(\bar{\delta})\nu_{0,0,0} = \bar{p}b(\ubar{\delta})\nu_{1,1,0}.\label{eq:prov-simplified-2}
\end{align}
Finally, solving~\eqref{eq:prov-total} and~\eqref{eq:prov-simplified-2} gives 
\begin{align}
    \nu_{0,0,0} &= \frac{\bar{p}b(\ubar{\delta})}{(1+a(1))\bar{p}b(\ubar{\delta})+(1+b(1))\bar{q}a(\bar{\delta})} =: \Gamma_0(a, b),\notag\\
    \nu_{1,1,0} &= \frac{\bar{q}a(\bar{\delta})}{(1+a(1))\bar{p}b(\ubar{\delta})+(1+b(1))\bar{q}a(\bar{\delta})} =: \Gamma_1(a, b).\notag
\end{align}
This completes the proof.

\renewcommand{\theequation}{D\arabic{equation}}
\setcounter{equation}{0} 

\section{Proof of Theorem~\ref{theorem:average-cost-threshold}}\label{proof:average-cost-threshold}
Recall that MC2 induced by the switching policy admits a unique stationary distribution. The average cost is simply the expected total cost of being in these states, i.e., $\sum_{i,j,\delta}\delta\nu_{i,j,\delta}$. The average cost of all MA errors is
\begin{align}
    &C_{\text{MA}}(\pi) =\beta\sum_{k=1}^{\infty} k\nu_{1,0,k}\notag\\
    &= \beta\bigg(\sum_{k=1}^{\bar{\delta}-1}kp\bar{q}^{k-1}\nu_{0,0,0} + \sum_{k=\bar{\delta}}^{\infty}kp\bar{q}^{\bar{\delta}-1}(\bar{q}p_f)^{k-\bar{\delta}}\nu_{0,0,0}\bigg)\notag\\
    &= \beta p \nu_{0,0,0}\bigg( \sum_{k=1}^{\bar{\delta}-1}k\bar{q}^{k-1} + \bar{q}^{\bar{\delta}-1}\sum_{k^\prime=0}^{\infty}(k^\prime + \bar{\delta})(\bar{q}p_f)^{k^\prime}\bigg).\label{eq:prov-ma-costs}
\end{align}
By using the facts that when $0<\rho<1$, $\sum_{k=0}^{\infty}k \rho^{k} = \frac{\rho}{(1-\rho)^2},~\sum_{k=1}^{n}k \rho^{k-1} = \frac{1-\rho^n}{(1-\rho)^2} - \frac{n\rho^n}{1-\rho}$, it follows that
\begin{align}
    C_{\text{MA}}(\pi) &= \beta p \nu_{0,0,0}
    \bigg(\frac{1-\big(\bar{q} + (1-\bar{q})\bar{\delta}\big)\bar{q}^{\bar{\delta}-1}}{(1-\bar{q})^2} \notag\\
    &\qquad \qquad \qquad + \frac{\big(\bar{q}p_f + (1-\bar{q}p_f)\bar{\delta}\big)\bar{q}^{\bar{\delta}-1}}{(1-\bar{q}p_f)^2}\bigg).\label{eq:prov-ma-costs-2}
 \end{align}
Similarly, the average cost of all FA errors is
\begin{align}
    C_{\text{FA}}(\pi) &= (1-\beta) q \nu_{1,1,0}\bigg(\frac{1-\big(\bar{p} + (1-\bar{p})\ubar{\delta}\big)\bar{p}^{\ubar{\delta}-1}}{(1-\bar{p})^2} \notag\\
    &\qquad \qquad \qquad + \frac{\big(\bar{p}p_f + (1-\bar{p}p_f)\ubar{\delta}\big)\bar{p}^{\ubar{\delta}-1}}{(1-\bar{p}p_f)^2}\bigg).\label{eq:prov-ma-costs-3}
 \end{align}
The transmission frequency can be computed by summing the steady probabilities of all exceeding-threshold states, i.e.,
\begin{align}
    F(\pi) \hspace{-0.2em}=\hspace{-0.2em} \sum_{k = \bar{\delta}}^{\infty}\hspace{-0.2em}\nu_{1,0,k} \hspace{-0.2em}+\hspace{-0.2em} \sum_{k = \ubar{\delta}}^{\infty}\hspace{-0.2em}\nu_{0,1,k} = a(\bar{\delta})\nu_{0,0,0} + b(\ubar{\delta})\nu_{1,1,0}.\label{eq:prov-ma-costs-4}
\end{align}
Theorem~\ref{theorem:average-cost-threshold} follows from~\eqref{eq:prov-ma-costs}-\eqref{eq:prov-ma-costs-4}.

\renewcommand{\theequation}{E\arabic{equation}}
\setcounter{equation}{0} 

\section{Proof of Proposition~\ref{proposition:truncated-MDP-convergence}}\label{proof:truncated-MDP-convergence}
We only verify the case when $\bar{\delta}, \ubar{\delta}\geq 1$. The proofs for the other cases are analogous. As depicted in Fig.~\ref{fig:truncated-DTMC}, the stationary probability of the boundary state $(1,0,N)$ satisfies 
\begin{align}
    \nu_{1,0,N}(N) = \bar{q}p_f \nu_{1,0,N-1}(N) + \bar{q}p_f \nu_{1,0,N}(N).\notag
\end{align}
Rearranging the term yields
\begin{align}
    \nu_{1,0,N}(N) = \frac{\bar{q}p_f}{1-\bar{q}p_f}\nu_{1,0,N-1}(N).
\end{align}
Then, the quantity $a_N(\bar{\delta})$ is given by
\begin{align}
a_N(\bar{\delta}) &= \sum_{k=\bar{\delta}}^{N}\frac{\nu_{1,0,k}(N)}{\nu_{0,0,0}(N)} \notag\\
&= \sum_{k=\bar{\delta}}^{N-1}p\bar{q}^{\bar{\delta}-1}(\bar{q}p_f)^{k-\bar{\delta}} + p\bar{q}^{\bar{\delta}-1}(\bar{q}p_f)^{N-\bar{\delta}-1}\frac{\bar{q}p_f}{1-\bar{q}p_f}\notag\\
&= p\bar{q}^{\bar{\delta}-1}\sum_{k^\prime=0}^{N-\bar{\delta}-1}(\bar{q}p_f)^{k^\prime} + \frac{p\bar{q}^{\bar{\delta}-1}(\bar{q}p_f)^{N-\bar{\delta}}}{1-\bar{q}p_f}\notag\\
&=p\bar{q}^{\bar{\delta}-1}\frac{1-(\bar{q}p_f)^{N-\bar{\delta}}}{1-\bar{q}p_f}+ \frac{p\bar{q}^{\bar{\delta}-1}(\bar{q}p_f)^{N-\bar{\delta}}}{1-\bar{q}p_f}\notag\\
&=\frac{p\bar{q}^{\bar{\delta}-1}}{1-\bar{q}p_f} = a(\bar{\delta}).\notag
\end{align}
Similarly, one can show that $a_N(1) = a(1)$, $b_N(1) = b(1)$, and $b_N(\ubar{\delta}) = b(\ubar{\delta})$. Therefore, $\nu_{i,j,\delta}(N)=\nu_{i,j,\delta}$ for all $(i,j,\delta)\in\mathcal{S}_N^\text{in}$. This yields~\eqref{eq:stationary-distribution-truncated}.

The average cost of all MA errors is given by
\begin{align}
    C_\text{MA}&(\pi,N) = \beta\sum_{k=1}^{N} k\nu_{1,0,k}(N) = \beta p\nu_{0,0,0} \bigg(\sum_{k=1}^{\bar{\delta}-1}k\bar{q}^{k-1}\notag\\
    &+ \sum_{k=\bar{\delta}}^{N-1}k\bar{q}^{\bar{\delta}-1}(\bar{q}p_f)^{k-\bar{\delta}}+ \frac{N\bar{q}^{\bar{\delta}-1}(\bar{q}p_f)^{N-\bar{\delta}}}{1-\bar{q}p_f} \bigg).\label{eq:prov-ma-total}
\end{align}
The first term of the rightmost expression is the same as in Eq.~\eqref{eq:prov-ma-costs}. The remaining terms can be computed as
\begin{align}
&\sum_{k=\bar{\delta}}^{N-1} k\bar{q}^{\bar{\delta}-1}(\bar{q}p_f)^{k-\bar{\delta}} + \frac{N\bar{q}^{\bar{\delta}-1}(\bar{q}p_f)^{N-\bar{\delta}}}{1-\bar{q}p_f} \notag\\
=~&\bar{q}^{\bar{\delta}-1}\bigg(\sum_{k^\prime=0}^{N-\bar{\delta}-1}(k^\prime+\bar{\delta})(\bar{q}p_f)^{k^\prime} + \frac{N(\bar{q}p_f)^{N-\bar{\delta}}}{1-\bar{q}p_f}\bigg) \notag\\
=~& \bar{q}^{\bar{\delta}-1}\bigg(
    \frac{\bar{q}p_f - (\bar{q}p_f)^{N-\bar{\delta}}}{(1-\bar{q}p_f)^2} - \frac{(N-\bar{\delta}-1)(\bar{q}p_f)^{N-\bar{\delta}}}{1-\bar{q}p_f} \notag\\
&\quad\quad\quad\quad\quad +\frac{\bar{\delta}(1-(\bar{q}p_f)^{N-\bar{\delta}})}{1-\bar{q}p_f} + \frac{N(\bar{q}p_f)^{N-\bar{\delta}}}{1-\bar{q}p_f}
    \bigg)\notag\\
=~& \frac{\big(\bar{q}p_f + (1-\bar{q}p_f)\bar{\delta}\big)\bar{q}^{\bar{\delta}-1}}{(1-\bar{q}p_f)^2} - \frac{\bar{q}^N p_f^{N-\bar{\delta}+1}}{(1-\bar{q}p_f)^2}.\label{eq:prov-second}
\end{align}
Substituting \eqref{eq:prov-ma-costs-2} and \eqref{eq:prov-second} into \eqref{eq:prov-ma-total} gives
\begin{align}
    C_\text{MA}(\pi,N) = C_\text{MA}(\pi) - \beta p\nu_{0,0,0}\frac{\bar{q}^N p_f^{N-\bar{\delta}+1}}{(1-\bar{q}p_f)^2}.\label{eq:prove-E-1}
\end{align}
Similarly, the average cost of all FA errors is
\begin{align}
    C_\text{FA}(\pi,N) = C_\text{FA}(\pi) - (1-\beta)q\nu_{1,1,0}\frac{\bar{p}^N p_f^{N-\ubar{\delta}+1}}{(1-\bar{p}p_f)^2}.\label{eq:prove-E-2}
\end{align}
Since $a_N(\bar{\delta}) = a(\bar{\delta})$ and $b_N(\ubar{\delta}) = b(\ubar{\delta})$, we have 
\begin{align}
    F(\pi,N) = a_N(\bar{\delta})\nu_{0,0,0} + b_N(\ubar{\delta})\nu_{1,1,0}= F(\pi).\label{eq:prove-E-3}
\end{align}
The result in~\eqref{eq:average-cost-threshold-truncated} follows from~\eqref{eq:prove-E-1}-\eqref{eq:prove-E-3}.

\renewcommand{\theequation}{F\arabic{equation}}
\setcounter{equation}{0} 

\section{Proof of Proposition~\ref{theorem:slice-method}}\label{proof:slice-method}
The average cost of a switching policy $\pi=(x,\ubar{\delta})$ is
\begin{align}
\mathcal{L}(x,\ubar{\delta})\hspace{-0.2em}=\hspace{-0.2em}\frac{\beta p \bar{p} b(\ubar{\delta}) \psi(x) \hspace{-0.2em}+\hspace{-0.2em} \big(\lambda (\bar{p}+\bar{q}) b(\ubar{\delta}) \hspace{-0.2em}+\hspace{-0.2em} (1\hspace{-0.2em}-\hspace{-0.2em}\beta)q\bar{q}\psi(\ubar{\delta})\big)a(x)}{\bar{p}b(\ubar{\delta})(1+\zeta(x)) + \big(\bar{q}(1+b(1))+\bar{p}b(\ubar{\delta})\big)a(x)},\notag
\end{align}
where $\psi(x) = \psi(\bar{q},x), \psi(\ubar{\delta}) = \psi(\bar{p},\ubar{\delta}), a(x) = \frac{p\bar{q}^{x-1}}{1-\bar{q}p_f}$, $\zeta(x) = \frac{p(1-\bar{q}^{x-1})}{1-\bar{q}}$. The derivatives of $\psi(x), a(x)$, and $\zeta(x)$ are 
\begin{align}
    a^{\prime}(x) &= \frac{p\bar{q}^{x-1}\ln{\bar{q}}}{1-\bar{q}p_f},\,\,\,
    \zeta^{\prime}(x) = \frac{-p\bar{q}^{x-1}\ln{\bar{q}}}{1-\bar{q}},\notag\\
    \psi^{\prime}(x) &=\frac{-(1-\bar{q})\bar{q}^{x-1} - \big(\bar{q} + (1-\bar{q})x\big)\bar{q}^{x-1}\ln{\bar{q}}}{(1-\bar{q})^2}\notag\\
    &\quad+ \frac{(1-\bar{q}p_f)\bar{q}^{x-1} + \big(\bar{q}p_f + (1-\bar{q}p_f)x\big) \bar{q}^{x-1}\ln{\bar{q}}}{(1-\bar{q}p_f)^2}\notag\\
    &= \bigg(\frac{-\bar{q}p_s(1+x\ln{\bar{q}})}{(1-\bar{q})(1-\bar{q}p_f)} \hspace{-0.2em}+ \hspace{-0.2em} \frac{\bar{q}p_f\ln{\bar{q}}}{(1-\bar{q}p_f)^2} \hspace{-0.2em}- \hspace{-0.2em} \frac{\bar{q}\ln{\bar{q}}}{(1-\bar{q})^2}\bigg)\bar{q}^{x-1}.\notag
\end{align}
For notational simplicity, let $\mathcal{L}(x,\ubar{\delta}) = g(x)/h(x)$, where
\begin{align}
    g(x) &= \beta p \bar{p} b(\ubar{\delta}) \psi(x) + \big(\lambda (\bar{p}+\bar{q}) b(\ubar{\delta}) + (1\hspace{-0.2em}-\hspace{-0.2em}\beta)q\bar{q}\psi(\ubar{\delta})\big)a(x),\notag\\
    h(x) &= \bar{p}b(\ubar{\delta})(1+\zeta(x)) + \big(\bar{q}(1+b(1))+\bar{p}b(\ubar{\delta})\big)a(x).\notag
\end{align}
The derivatives of $g(x)$ and $h(x)$ are given by
\begin{align}
    g^{\prime}(x) = (\alpha_0 + \alpha_1 x) \bar{q}^{x-1},\quad h^{\prime}(x) = \alpha_2 \bar{q}^{x-1}.\notag
\end{align}
where $\alpha_i, i=0, 1, 2$, are obtained as
\begin{align}
    \alpha_0 &= \bigg(\lambda (\bar{p}+\bar{q}) b(\ubar{\delta}) + (1-\beta)q\bar{q}\psi(\ubar{\delta})\bigg)\frac{p\ln{\bar{q}}}{1-\bar{q}p_f} +\beta p \bar{p} b(\ubar{\delta}) \notag\\
    &\times
    \bigg(\frac{-\bar{q}p_s}{(1-\bar{q})(1-\bar{q}p_f)} + \frac{\bar{q}p_f\ln{\bar{q}}}{(1-\bar{q}p_f)^2} - \frac{\bar{q}\ln{\bar{q}}}{(1-\bar{q})^2}\bigg),\label{eq:prove-F-1}\\
    \alpha_1 &= \frac{-\beta p \bar{p} b(\ubar{\delta})\bar{q}p_s \ln{\bar{q}}}{(1-\bar{q})(1-\bar{q}p_f)}>0,\label{eq:prove-F-2}\\
    \alpha_2 &= \bigg(\frac{\bar{p}p_s b(\ubar{\delta})}{(1-\bar{q})(1-\bar{q}p_f)}+\frac{\bar{q}(1+b(1)}{1-\bar{q}p_f}\bigg)\ln{\bar{q}}<0.\label{eq:prove-F-3}
\end{align}
The first term on the right-hand side of~\eqref{eq:prove-F-1} is negative, while the second term can be either positive or negative. Hence, the sign of $g^{\prime}(x)$ is uncertain.

With~\eqref{eq:prove-F-1}-\eqref{eq:prove-F-3}, the derivative of $\mathcal{L}(x, \ubar{\delta})$ is given by
\begin{align}
    \frac{\rm{d}\mathcal{L}(x, \ubar{\delta})}{\rm{d} x} &= \frac{g^{\prime}(x) h(x) - h^{\prime}(x) g(x)}{h^2(x)} \notag\\
    &=\frac{\big((\alpha_0 + \alpha_1 x) h(x) - \alpha_2 g(x)\big)\bar{q}^{x-1}}{h^2(x)}.
\end{align}
Since $g(x)$ and $h(x)$ are positive-valued functions, it follows that when $\alpha_0\geq 0$, $\frac{\rm{d}\mathcal{L}(x, \ubar{\delta})}{\rm{d} x} > 0$ holds for all $x\geq 1$. In this case, $\mathcal{L}(x, \ubar{\delta})$ is a monotone increasing function and $x^*_{\ubar{\delta}} = 1$. When $\alpha_0 < 0$, it follows that $\frac{\rm{d}\mathcal{L}(x, \ubar{\delta})}{\rm{d} x} > 0$ for all $x\geq x^*_{\ubar{\delta}}$, where $x^*_{\ubar{\delta}} = -\frac{\alpha_0}{\alpha_1}$. Similarly, one can show that, for each $\bar{\delta}\geq 1$, there exists some $y^*_{\bar{\delta}} \geq 1$ such that $\mathcal{L}(\bar{\delta}, y)$ is increasing for all $y \geq y^*_{\bar{\delta}}$. This completes our proof.

\bibliographystyle{IEEEtran}
\bibliography{ref}

\end{document}